\documentclass[11pt, letterpaper]{article}
\usepackage[utf8]{inputenc}
\usepackage{fullpage,times}

\usepackage[procnumbered,ruled,vlined,linesnumbered]{algorithm2e}

\DontPrintSemicolon
\SetKw{KwAnd}{and}
\SetProcNameSty{textsc}
\SetFuncSty{textsc}
\usepackage{amsmath,amsfonts,amsthm,amssymb,multirow}
\usepackage{times}
\usepackage{graphicx}
\usepackage{floatpag}
\usepackage{bbold}
\usepackage{enumitem}
\usepackage{subcaption} 
\usepackage{calc}
\usepackage{tikz}
\usepackage{hyperref}
\usetikzlibrary{decorations.markings}
\tikzstyle{vertex}=[circle, draw, inner sep=0pt, minimum size=4pt, fill = black]

\usepackage{graphicx}
\usepackage{tabularx}
\makeatletter
\newcommand{\multiline}[1]{%
  \begin{tabularx}{\dimexpr\linewidth-\ALG@thistlm}[t]{@{}X@{}}
    #1
  \end{tabularx}
}
\makeatother
\usepackage{mathtools}

\makeatletter
\def\BState{\State\hskip-\ALG@thistlm}
\makeatother

\usepackage[compact]{titlesec}
\titlespacing{\section}{0pt}{3ex}{2ex}
\titlespacing{\subsection}{0pt}{2ex}{1ex}
\titlespacing{\subsubsection}{0pt}{0.5ex}{0ex}

\newtheorem{theorem}{Theorem}[section]

\newtheorem{corollary}{Corollary}[section]

\newtheorem{definition}{Definition}[section]
\newtheorem{lemma}{Lemma}[section]
\newtheorem{claim}{Claim}
\newtheorem{hypothesis}{Hypothesis}

\newtheorem{observation}{Observation}[section]

\makeatletter
\let\c@fconjecture\c@conjecture
\makeatother

\makeatletter
\let\c@fconj\c@conj
\makeatother

\def \eps {\varepsilon}

\newcommand{\ignore}[1]{}

\def \poly { \text{\rm poly~} }

\def\tO{\tilde{O}}

\title{Tight Conditional Lower Bounds for Approximating Diameter in Directed Graphs}
\author{Mina Dalirrooyfard\footnote{Supported by NSF CAREER Award, NSF Grants CCF-1740519 and CCF-1909429. minad@mit.edu}\\MIT \and  Nicole Wein\footnote{Supported by NSF Grant CCF-1514339. nwein@mit.edu}\\MIT}
\date{}
\begin{document}

\maketitle
\thispagestyle{empty}

\begin{abstract}
Among the most fundamental graph parameters is the Diameter, the largest distance between any pair of vertices in a graph. Computing the Diameter of a graph with $m$ edges requires $m^{2-o(1)}$ time under the Strong Exponential Time Hypothesis (SETH), which can be prohibitive for very large graphs, so efficient \emph{approximation} algorithms for Diameter are desired. 

There is a folklore algorithm that gives a $2$-approximation for Diameter in $\tilde{O}(m)$ time (where $\tilde{O}$ notation suppresses logarithmic factors). Additionally, a line of work [SODA'96, STOC'13, SODA'14] concludes with a $3/2$-approximation algorithm for Diameter in weighted directed graphs that runs in $\tilde{O}(m^{3/2})$ time. For directed graphs, these are the only known approximation algorithms for Diameter.

The $3/2$-approximation algorithm is known to be tight under SETH: Roditty and Vassilevska W. [STOC'13] proved that under SETH any $3/2-\varepsilon$ approximation algorithm for Diameter in undirected unweighted graphs requires $m^{2-o(1)}$ time, and then Backurs, Roditty, Segal, Vassilevska W., and Wein [STOC'18] and the follow-up work of Li proved that under SETH any $5/3-\varepsilon$ approximation algorithm for Diameter in undirected unweighted graphs requires $m^{3/2-o(1)}$ time. 

Whether or not the folklore 2-approximation algorithm is tight, however, is unknown, and has been explicitly posed as an open problem in numerous papers. Towards this question, Bonnet recently proved that under SETH, any $7/4-\varepsilon$ approximation requires $m^{4/3-o(1)}$, only for directed weighted graphs. 

We completely resolve this question for directed graphs by proving that the folklore 2-approximation algorithm is conditionally optimal. In doing so, we obtain a series of conditional lower bounds that together with prior work, give a complete time-accuracy trade-off that is tight with all known algorithms for directed graphs. Specifically, we prove that under SETH for any $\delta>0$, a $(\frac{2k-1}{k}-\delta)$-approximation algorithm for Diameter on directed unweighted graphs requires $m^{\frac{k}{k-1}-o(1)}$ time.

\end{abstract}
\vfill
\pagebreak
\section{Introduction}

 Among the most fundamental graph parameters is the Diameter, the largest distance between any pair of vertices in a graph i.e. $\max_{u,v\in V}d(u,v)$, where $V$ is the vertex set. Efficient algorithms for computing Diameter are sought after in practice~\cite{brandes2005network,diam-prac5,diam-prac4, diam-prac2,diam-prac6,diam-prac3,ceccarello2020distributed}. From the theoretical side, algorithms for computing Diameter have been studied in a wide variety of contexts such as in the distributed~\cite{peleg2012distributed,frischknecht2012networks,holzer2014brief,grossman2020improved,kuhn2020computing}, dynamic~\cite{ancona2019algorithms,van2019dynamic}, parameterized~\cite{abboud2016approximation,bentert2019parameterized}, and quantum~\cite{le2018sublinear} settings. A number of variants of Diameter under different distance measures have also recently been proposed and studied~\cite{abboud2016approximation,dalirrooyfard2019approximation,dalirrooyfard2019tight}. We study the standard version of Diameter, which is one of the central problems in \emph{fine-grained complexity}. 
 
 The fastest known algorithms \cite{ryanapsp,PettieR05,Pettie04} for Diameter in $n$-vertex $m$-edge graphs are only slightly faster (by $n^{o(1)}$ factors) than the simple $\tilde{O}(mn)$ time\footnote{$\tilde{O}$ notation supresses polylogarithmic factors.} algorithm of running Dijkstra's algorithm from every vertex and then taking the largest distance. For dense graphs with small integer weights there are improved algorithms \cite{seidel,Zwick02,cyganbaur} using fast matrix multiplication, but these algorithms are not faster than $mn$ for sparser graphs or graphs with large weights. Furthermore, under the Strong Exponential Time Hypothesis (SETH), there is no $O(m^{2-\varepsilon})$ time algorithm for any constant $\varepsilon>0$ for computing Diameter even in unweighted, undirected graphs~\cite{RV13}. Since quadratic time can be prohibitively slow on very large graphs, finding efficient \emph{approximation} algorithms for Diameter is desirable.

A folklore $\tilde{O}(m)$ algorithm gives a 2-approximation for Diameter in directed weighted graphs. The algorithm simply picks an arbitrary vertex $v$, runs Dijkstra's algorithm from $v$ (in both directions if the graph is directed), and returns the largest distance found. By the triangle inequality, the value returned is at least half of the true diameter. The first non-trivial approximation algorithm for Diameter was by Aingworth, Chekuri, Indyk, and Motwani ~\cite{aingworth}, who presented an almost-$3/2$-approximation\footnote{An almost-$c$-approximation of $X$ is an estimate $X'$ so that $X\leq X'\leq cX+O(1)$.} algorithm for Diameter in unweighted directed graphs running in $\tilde{O}(n^2+m\sqrt n)$ time. Roditty and Vassilevska W.~\cite{RV13} then improved the running time to $\tilde{O}(m\sqrt n)$ in expectation. This was extended in~\cite{ChechikLRSTW14}
%Chechik, Larkin, Roditty, Schoenebeck, Tarjan, and Vassilevska W.
to obtain a (genuine) $3/2$-approximation algorithm for Diameter in weighted directed graphs running in $\tilde{O}(\min \{m^{3/2},mn^{2/3}\})$ time. Cairo, Grossi, and Rizzi~\cite{cairo} generalized the above results for undirected graphs with small weights and obtained a time-accuracy trade-off: for every $k\geq 1$ they obtained an $\tO(mn^{1/(k+1)})$  time algorithm that achieves an almost-$2-1/2^k$-approximation.

The above $3/2$-approximation algorithm in $\tilde{O}(m^{3/2})$ time is conditionally tight for sparse graphs in terms of both its approximation factor and its running time~\cite{RV13,stoc2018}. In particular, Roditty and Vassilevska W.~\cite{RV13} proved that under SETH, any $3/2-\varepsilon$ approximation algorithm (for $\varepsilon>0$) for Diameter in undirected unweighted graphs requires $m^{2-o(1)}$ time. In STOC'18, Backurs, Roditty, Segal, Vassilevska W., and Wein~\cite{stoc2018} proved that under SETH, any $8/5-\varepsilon$ approximation algorithm for Diameter in undirected unweighted graphs, or any $5/3-\varepsilon$ approximation for Diameter in undirected weighted graphs requires $m^{3/2-o(1)}$ time. 

Although the $3/2$-approximation algorithm is conditionally tight, the tightness of the folklore 2-approximation algorithm remains completely unclear. We focus on the following question, which was asked as Open Question 2.2 in the survey~\cite{rubinstein2019seth} by Rubinstein and Vassilevska W., and has also been explicitly asked in several other works~\cite{stoc2018,4ov}.

\begin{center}
{\bf Main Question:} \emph{Is the folklore $\tilde{O}(m)$ time 2-approximation algorithm for Diameter optimal?} 
\end{center}

Notably, for the related problem of Eccentricities, where the goal is to find the largest distance from every vertex in the graph, there is an analogous folklore algorithm and an analogous Main Question, which was resolved in~\cite{stoc2018}. They gave an $\tilde{O}(m)$ algorithm for Eccentricities that improves over the folklore algorithm, and showed that it is conditionally tight. However, their algorithm does not carry over to Diameter, and their hardness constructions only partially carry over to Diameter, so Diameter remained elusive. 

Recently, there has been some progress on the Main Question for Diameter. Li~\cite{3vs5} (in a earlier version of his paper) improved the unweighted undirected construction of~\cite{stoc2018} to match the weighted undirected construction of~\cite{stoc2018}. That is, he showed that under SETH, any $5/3-\varepsilon$ approximation algorithm for Diameter in undirected unweighted graphs requires $m^{3/2-o(1)}$ time. Then, Bonnet~\cite{4ov} surpassed this $5/3$ bound for \emph{directed weighted} graphs, by showing that under SETH, any $7/4-\varepsilon$ approximation algorithm requires $m^{4/3-o(1)}$ time. That is, before this work, there was a gap between $7/4$ and $2$ for the optimal approximation factor for an $\tilde{O}(m)$-time algorithm for Diameter in directed weighted graphs, and a gap between $5/3$ and $2$ for undirected unweighted graphs.

\subsection{Our results}

We completely resolve the Main Question in the affirmative for directed graphs. We obtain a series of conditional lower bounds that give a full time-accuracy trade-off and show that the folklore $\tilde{O}(m)$ time 2-approximation algorithm for Diameter is optimal under SETH for directed graphs. Moreover, we improve the result of Bonnet \cite{4ov} to hold for \emph{unweighted directed} graphs.

Specifically, we prove the following theorem. See Figure~\ref{fig:results} for a plot of prior work and our results.
\begin{theorem}
\label{thm:main}
Let $k\ge 4$ be a fixed integer. Assuming SETH, for all $\delta>0$, any $(\frac{2k-1}{k}-\delta)$-approximation algorithm for Diameter in an unweighted directed graph on $m$ edges requires $m^{\frac{k}{k-1}-o(1)}$ time.
\end{theorem}

\begin{figure}[h]
    \centering
    \includegraphics{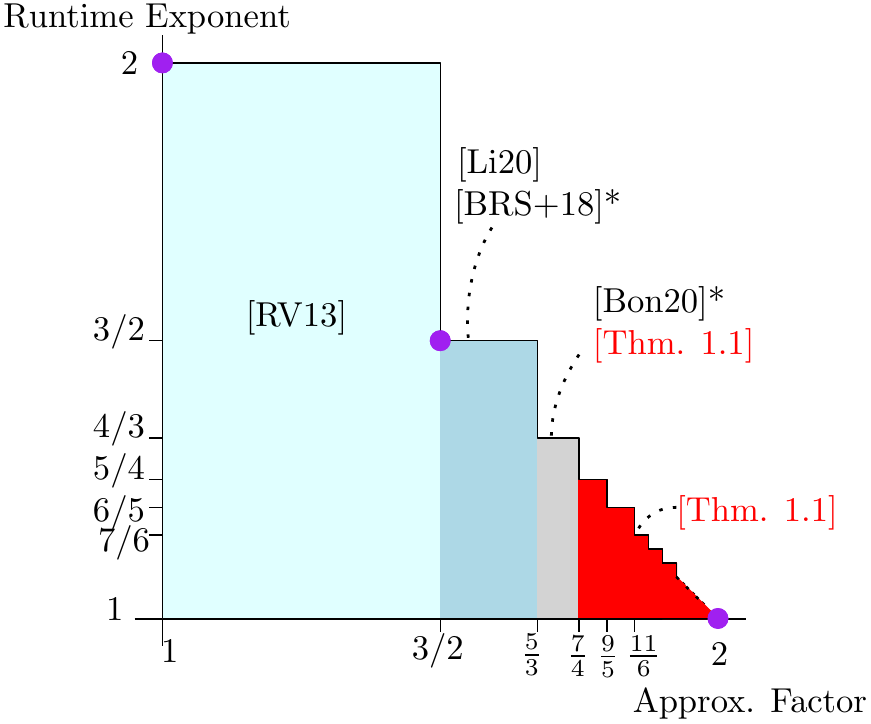}
    \caption{Runtime exponent versus approximation factor for Diameter in unweighted directed graphs. The grey and red areas are our new results which are summarized in Theorem \ref{thm:main}. The purple dots are existing algorithms. \\ *Lower bounds that were proved for weighted graphs, and later improved to hold for unweighted graphs.}
    \label{fig:results}
\end{figure}

Together with prior work, our result gives a complete time-accuracy trade-off that is tight with all known algorithms for directed graphs. In particular, combining our result with prior work, we have the following corollary, which is identical to Theorem~\ref{thm:main}, except with $k\ge 2$ instead of 4.

\begin{corollary}
Let $k\ge 2$ be a fixed integer. Assuming SETH, for all $\delta>0$, any $(\frac{2k-1}{k}-\delta)$-approximation algorithm for Diameter in an unweighted directed graph on $m$ edges requires $m^{\frac{k}{k-1}-o(1)}$ time.
\end{corollary}

\paragraph{Concurrent work by Li \cite{3vs5}} In an independent and concurrent work, Li proved the same result as Theorem \ref{thm:main} with a different proof. 
\subsection{Our Techniques}

Backurs, Roditty, Segal, Vassilevska W., and Wein \cite{stoc2018} define a variant of Diameter called $ST$-Diameter where given a graph $G=(V,E)$, and subsets $S,T\subseteq V$ the problem asks for $\max_{s\in S,t\in T}d(s,t)$. % which is defined as the $ST$ diameter of $G$. 
 They provide the following time-accuracy trade-off lower bounds for $ST$-diameter. They prove that under SETH, for every $k\ge 2$, every algorithm that can distinguish between $ST$-diameter $k$ and $3k-2$ in undirected unweighted graphs requires $n^{\frac{k}{k-1}-o(1)}$ time. Given a $k$-OV instance, they construct a graph $G$ with $k+1$ layers, letting the first layer be $S$ and the last layer be $T$, such that the $ST$-diameter is $k$ if the $k$-OV instance is a NO instance, and the $ST$-diameter is $3k-2$ if the $k$-OV instance is a YES instance. 
 
 Each node in $S$ and $T$ represents a $k-1$ tuple $(v_1,\ldots,v_{k-1})$ of vectors from the $k$-OV instance, and the rest of the nodes represent a $k-1$ tuple $(v_1,\ldots,v_{k-2},x)$ where each $v_i$ is again a vector from the $k$-OV instance, and  $x=(x_1,\ldots,x_{k-2})$ is an array where $x_i\in [d]$. 
A path between $(v_1,\ldots,v_{k-1})\in S$ and $(u_1,\ldots,u_{k-1})\in T$ passes through all the layers, and each edge changes one of the $v_i$ vectors to a $u_j$ vector in a particular order, and ensures that these vectors have 1s at particular entries specified by the array $x$. Like all of the work on Diameter lower bounds subsequent to \cite{stoc2018}, we use this $ST$-Diameter construction as a starting point. 

The construction of \cite{stoc2018} does not directly work for Diameter since in the NO case the diameter of the graph might be as big as $2k$, as two vertices in $S$ can be far from each other. So, to get a construction for Diameter, the challenge is to add more vertices and edges to make these $S$-$S$ distances smaller in the NO case, while not decreasing the $S$-$T$ distances in the YES case.

To address this challenge for the case of $k=4$, Bonnet \cite{4ov} uses the following key idea. He copies one of the middle layers, where edges between this layer and its copy allow you to change the vectors of a node $(v_1,v_2,x)$. 
 These extra edges allow for shorter paths in the NO case exclusively.

A natural way to generalize Bonnet's construction to larger values of $k$ is to make a copy of each of the internal layers. However, it is not clear which of the vectors we should allow to change on an edge from a layer to its copy. If we allow certain vectors to change at the wrong layer, this shrinks the diameter too much in the YES case, while if we don't allow enough flexibility, this does not adequately decrease the distances in the NO case. A key insight for our construction is that  changing different sets of vectors should have different costs. 

These costs are encoded in our construction through an intricate system of ``back edges''. These back edges allow for paths that go from some layer to some previous layer while changing either some prefix or some suffix of the vector tuple. The size of the prefix or suffix that is permitted to change depends on which pair of layers these back edges are connecting. By carefully balancing the number of vectors we are permitted to change for which pairs of levels, we manage to keep all the distances at most $k$ in the NO case, and ensure that the diameter is at least $2k-1$ in the YES case.

\subsection{Organization}
In Section \ref{sec:prelim} we provide some preliminaries. 
Sections \ref{sec:k>5}, \ref{sec:k>5no} and \ref{sec:k>5yes} are devoted to Theorem \ref{thm:main} for the case where $k\ge 5$. In Section \ref{sec:k>5} we introduce the construction of our reduction, and in Sections \ref{sec:k>5no} and \ref{sec:k>5yes} we prove both directions of the reduction, respectively. Finally in Section \ref{sec:k=4} we prove Theorem~\ref{thm:main} for the case where $k=4$.
%In section \ref{sec:k>5} 

\section{Preliminaries}
\label{sec:prelim}

Let $G = (V, E)$ be a directed graph, where $n=|V|$ and $m=|E|$. For every $u,v \in V$ let $d(u,v)$ be the length of the shortest path from $u$ to $v$.

Let $k\ge 2$. The \textit{$k$-Orthogonal Vectors Problem ($k$-OV)} is as follows: Given a set $S$ of $n$ vectors in $\{0,1\}^d$, determine whether there exist $v_1,\ldots,v_k\in S$ so that their generalized inner product is 0, i.e. $\sum_{i=1}^d \prod_{j=1}^k v_j[i]=0$, where $v_j[i]$ is the $i$th bit of the vector $v_j$.

Our conditional lower
bounds are based on the \textit{$k$-OV Hypothesis}, defined as below:

\begin{hypothesis}[$k$-OV Hypothesis]\label{conj:kov} For all constants $k\ge2$: there exists $c_k>0$ such that $k$-OV on $d=c_k\log{n}$ bit vectors requires $n^{k-o(1)}$ time on a
word-RAM with $O(\log{n})$ bit words.
\end{hypothesis}
Williams \cite{TCS05} showed that if the $k$-OV Hypothesis is false, then CNF-SAT on formulas with $N$ variables and $m$ clauses can be solved in $2^{N(1-\varepsilon/k)}\poly(m)$ time. 
In particular, such an algorithm would contradict the Strong Exponential Time Hypothesis (SETH) of Impagliazzo, Paturi and Zane \cite{ipz2} which is the following: For every $\eps>0$ there is a $K$ such that $K$-SAT on $N$ variables cannot be solved in $2^{(1-\varepsilon)N}\poly(N)$ time (say, on a word-RAM with $O(\log{N})$ bit words). This means that SETH implies the $k$-OV Hypothesis.

\section{The Construction}
\label{sec:k>5}
%\nicole{should be $\geq$}
In this section we suppose that $k\ge 5$ and we prove Theorem \ref{thm:main} by reduction from $k$-OV. We are given a $k$-OV instance $S$ where each vector in $S$ is of length $d=c_k\log{n}$, where $c_k$ is the constant defined in Hypothesis \ref{conj:kov}. %of $d$-bit numbers

We will create a graph $G=(V,E)$ with $O(n^{k-1}+n^{k-2}d^{k-1})$ vertices and $O(n^{k-1}d^{2k-2})$ edges, such that if the $k$-OV instance is a YES instance then $G$ has diameter $2k-1$, and if it is a NO instance then $G$ has diameter $k$. 

Before presenting the construction of $G$, we note that the above conditions on $G$ suffice to prove Theorem~\ref{thm:main}: suppose for contradiction that there exist $\delta>0$ and $\varepsilon>0$, such that there is a $(\frac{2k-1}{k}-\delta)$-approximation algorithm $\mathcal{A}$ for Diameter in directed graphs with $M$ edges that runs in $O(M^{\frac{k}{k-1}-\varepsilon})$ time. That is, $\mathcal{A}$ can distinguish whether $G$ has diameter $k$ or $2k-1$ in $O(|E|^{\frac{k}{k-1}-\varepsilon})=\tilde{O}(n^{k-(k-1)\varepsilon})$ time, where the last equality comes from the fact that $|E|=\tilde{O}(n^{k-1})$, since $k$ is constant and $d=O(\log{n})$. Then, the reduction from $k$-OV to Diameter on the graph $G$ tells us that we can solve the $k$-OV instance in $\tilde{O}(n^{k-(k-1)\varepsilon})$ time, which contradicts the $k$-OV Hypothesis.

\subsection{Vertex set}
Given our $k$-OV instance, we first augment $S$ with the all 1s vector. Note that this does not change the output of the $k$-OV instance. Then, we make $k$ copies of $S$ and call them $S_1,\dots,S_{k}$. We let a \emph{coordinate} be an element of $[d]$, that will represent a position in some vector in $S_1,\dots,S_{k}$. 

We start with $k+1$ layers of vertices $L_1,\dots,L_{k+1}$. Vertices of $L_1$ are $k-1$ tuples $(a_1,\dots,a_{k-1})$, with $a_i\in S_i$. Vertices of $L_{k+1}$ are $k-1$ tuples $(b_2,\dots,b_{k})$, where $b_i\in S_i$. For $i=2,\dots,k$, the vertices of $L_i$ are $(a_1,\dots, a_{k-i},b_{k-i+3},\dots ,b_{k},x)$, where $a_j,b_j\in S_j$ for each $j$ and $x=(x_1,\dots,x_{k-1})$ is an array of coordinates that satisfies the following two conditions: (1) For each $1\le j\le k-i$, 
$a_j[x_{\ell}]=1$ for all $1\le \ell \le k-j$, and (2) For each $k-i+3\le j\le k$,
$b_j[x_\ell]=1$ for all $k-j+1\le \ell \le k-1$. See table \ref{tab:Li}.

\begin{table}[!htb]
      \centering
         \begin{tabular}{|c|c |c|c|c|c| c|c| c| c|} 
 \hline
 & $x_1$ &$x_2$ & $\ldots$ & $x_{i-2}$&$x_{i-1}$& $x_{i}$ & $\ldots$& $x_{k-2}$& $x_{k-1}$\\ \hline
 $a_1$ & $1$ &$1$ & $\ldots$ & $1$ & $1$& $1$&$\ldots$ & $1$ & $1$ \\ \hline
 $a_2$ & $1$ &$1$ & $\ldots$ & $1$ & $1$& $1$& $\ldots$ & $1$ &  \\ \hline
 $\ldots$ & $\ldots$ &$\ldots$ & $\ldots$ & $\ldots$ & $\ldots$ & $\ldots$ & $\ldots$&&\\ \hline
 $a_{k-i}$ & $1$&$1$ & $\ldots$ & $1$&  $1$ &$1$ & & & \\ \hline
  $b_{k-i+3}$ & & &  & $1$& $1$& $1$ &$\ldots$&$1$&$1$\\ \hline
 $\ldots$ &  &  & $\ldots$ & $\ldots$ & $\ldots$ &$\ldots$&$\ldots$&$\ldots$&$\ldots$\\ \hline
 $b_{k-1}$ &  &  $1$ & $\ldots$& $1$&$1$ & $1$ &$\ldots$&$1$&$1$\\ \hline

 $b_{k}$ & $1$  &  $1$ & $\ldots$& $1$&$1$ & $1$ &$\ldots$&$1$&$1$ \\ \hline
\end{tabular}

    \caption{The relationship between the vector array $a_1,\ldots,a_{k-i},b_{k-i+3},\ldots,b_k$ and the coordinate array $x$ of a node $(a_1,\ldots,a_{k-i},b_{k-i+3},\ldots,b_k,x)$ in layer $L_i$. }
    \label{tab:Li}
\end{table}

For each $i=3,\dots,k-1$, we add a set $L'_i$ of vertices.
Vertices of $L'_i$ are $(a_1,\dots, a_{k-i},b_{k-i+3},\dots ,b_{k},x)$ where $a_j,b_j\in S_j$ and $x=(x_1,\ldots,x_{k-1})$ is any coordinate array. 

For $i=4,\ldots,k-1$, we have a set $A_i$, where each node $\alpha\in A_i$ is a $k-i$ tuple of vectors $(a_1,\ldots,a_{k-i})$, with $a_{\ell}\in S_{\ell}$. 
For $i=3,\ldots,k-2$, we have a set $B_i$, where each node $\alpha\in B_i$ is a $i-2$ tuple of vectors $(b_{k-i+3},\ldots,b_k)$, with $b_{\ell}\in S_{\ell}$.  Note that these nodes do not have a coordinate array.

Let $A=\cup_i A_i$, let $B=\cup_i B_i$, let $L=\cup_i L_i$, and let $L'=\cup_i L'_i$. For all $i$, let \emph{level $i$} denote $L_i\cup L'_i\cup A_i \cup B_i$ (or the union of these sets that exist for that $i$). In contrast, we use the word \emph{layer} to refer to an individual set $L_i$ or $L'_i$. 

Finally we have 2 additional vertices, $u$ and $v$.
This completes the definition of the vertex set of $G$. See Figure \ref{fig:fixed_edges}. Figure \ref{fig:fixed_edges} does not capture the case of $k=5$ since in this case $L_{k-2}$ comes before $L_4$, so we include Figure \ref{fig:k5} to depict the $k=5$ case. 
\paragraph{Number of nodes:} The number of nodes of layers $L_1$ and $L_{k+1}$ is $n^{k-1}$, and the number of nodes in each layer $L_i$ and $L'_i$ for $1<i<k+1$ is at most $n^{k-2}d^{k-1}$. The number of nodes in each $A_i$ and $B_i$ is at most $n^{k-2}$ and the number of fixed nodes is constant, so the graph has $O(n^{k-1}+n^{k-2}d^{k-1})$ nodes.

\subsection{Edge set}
All edges are undirected unless otherwise specified. We have five types of edges: \textit{fixed edges},
\textit{coordinate-change edges}, \textit{vector-change edges}, \textit{swap edges} and \textit{back edges}.

\begin{itemize}
\item A \emph{fixed edge} has $u$ or $v$ as one endpoint and is directed.  
\item A \emph{coordinate-change edge} is between two nodes having the same sequence of vectors and different coordinate arrays and is undirected.
\item A \emph{vector-change edge} is between two nodes with the same coordinate array and the same vector array except for at most one entry, where a vector in some $S_i$ is changed for another vector in $S_i$. A vector-change edge is undirected.
\item A \emph{swap edge} is between two nodes with the same coordinate array and the same vector array except for one entry, where a vector in $S_i$ is changed for a vector in $S_{i+2}$, or vice versa. A swap edge can also be between $L_1$ and $L_2$ or between $L_{k}$ and $L_{k+1}$, in which case a vector is changed for a coordinate array, or vice versa. A swap edge is undirected.
\item A \emph{back edge} is an edge with at least one endpoint in $A$ or $B$, and is directed. A back edge incident to one vertex in $B$ and one vertex in $A$ is called a $ba$-type back edge. Otherwise, a back edge is called an $a$-type back edge if it is incident to a vertex in $A$, and a $b$-type back edge if it is incident to a vertex in $B$.
\end{itemize}
Now we specify each of these edges in the graph.
\paragraph{Swap edges:} 
For each $i=2,\dots,k-1$, there are swap edges between $(a_1,\dots, a_{k-i},b_{k-i+3},\dots ,b_{k},x)\in L_i$ and $(a_1,\dots, a_{k-i-1},b_{k-i+2},\dots ,b_{k},x)\in L_{i+1}$. There are also swap edges between $(a_1,\dots,a_{k-1})\in L_1$ and $(a_1,\dots,a_{k-2},x)\in L_2$, as well as between $(b_2,\dots,b_{k})\in L_{k+1}$ and $(b_3,\dots,b_{k},x)\in L_k$.  

\paragraph{Vector-change edges:} For each $i=3,\dots,k-1$, there are vector-change edges between $L_i$ and $L'_i$. These edges are between $\alpha=(a_1,\dots, a_{k-i},b_{k-i+3},\dots ,b_{k},x)\in L_i$ and $\beta\in L'_i$, where $\beta$ has coordinate array equal to $x$, and the same vectors as $\alpha$, except for at most one of $a_{k-i}$ or $b_{k-i+3}$. %$A_{k-i}$ or $A_{k-i+3}$. 

\paragraph{Coordinate-change edges:} For each $i=3,\dots,k-1$, there are coordinate-change edges within each $L'_i$, and between $L'_i$ and $L_i$. 
We also have coordinate-change edges within $L_2$ and $L_{k}$. 

\paragraph{Back edges:} For $i=4,\ldots,k-1$, and for every $\alpha=(a_1,\ldots,a_{k-i},b_{k-i+3},\ldots,b_k,x)\in L_i\cup L'_i$,
we add an $a$-type back edge from $\alpha$ to $\beta=(a_1,\ldots,a_{k-i})\in A_i$. For every node $\beta=(a_1,\ldots,a_{k-i})\in A_i$, we add an $a$-type back edge from $\beta$ to any vertex $(c_1,\ldots,c_{k-4},c_{k-1},c_k,x)\in L'_4$, if $c_j=a_j$ for all $j=1,\ldots,k-i$. For $i=3,\ldots,k-2$ and every node $\alpha=(a_1,\ldots,a_{k-i},b_{k-i+3},\ldots,b_k,x)\in L_i\cup L'_i$, we add a $b$-type back edge from $\beta=(b_{k-i+3},\ldots,b_k)\in B_i$ to $\alpha$. We add a $b$-type back edge from every node $(c_1,c_2,c_5,\ldots,c_k,x)\in L'_{k-2}$ to $\beta=(b_{k-i+3},\ldots,b_k)\in B_i$ if $c_j=b_j$ for every $j=k-i+3,\ldots,k$. See Figure \ref{fig:Li}. 

Additionally, for any $i=3,\ldots,k-2$, we add a $b$-type back edge from $(a_{k-i+3},\ldots,a_k)\in B_i$ to $(a_k)\in B_{3}$. For any $i=4,\ldots,k-1$, we add an $a$-type back edge from $(a_1)\in A_{k-1}$ to $(a_1,\ldots,a_{k-i})\in A_i$. Also, for any $i=4,\dots,k-2$, we add a $ba$-type back edge from every node in $B_i$ to every node in $A_i$. Note that for $k=5$, we don't have any $ba$ type back edges since $A=A_4$ and $B=B_3$. See figure \ref{fig:k5}.

\begin{figure}[h]
  \centering
  \includegraphics[width=\linewidth]{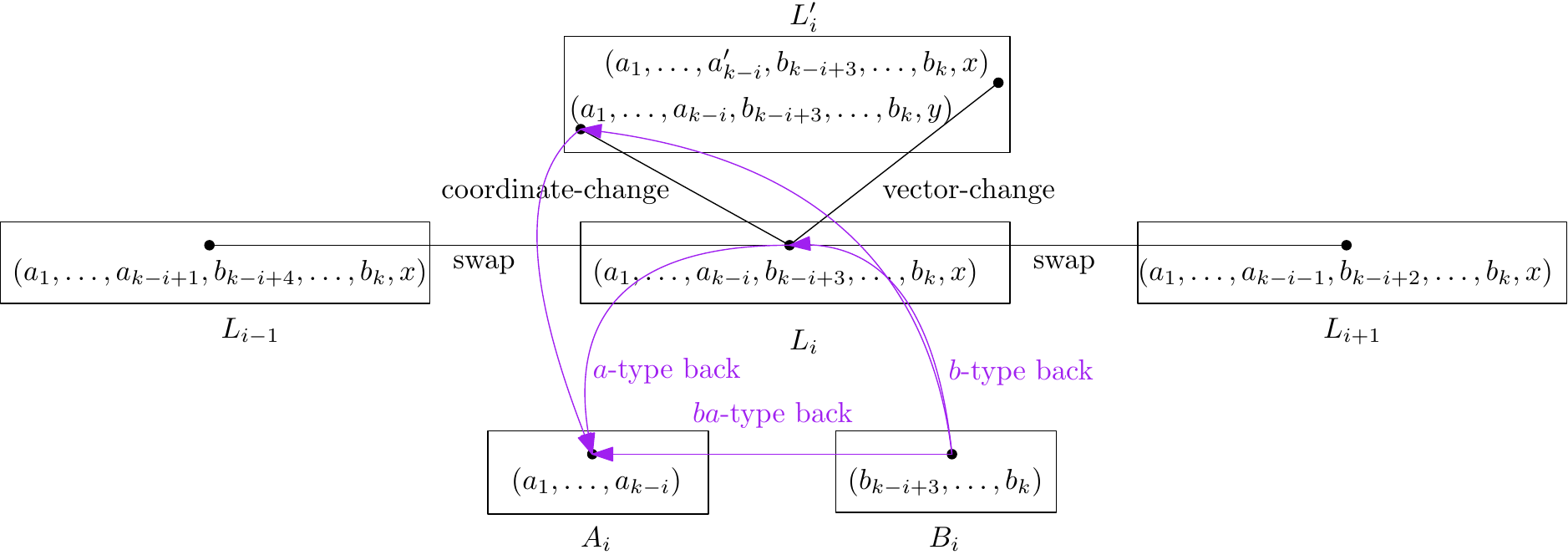}
  \caption{Edges attached to a node in $L_i$, for $i=4,\ldots,k-2$. Purple edges are back edges.}
  \label{fig:Li}
\end{figure}

\paragraph{Fixed edges:} Now we specify fixed edges. There is a directed edge from each vertex of $L'_{k-2}$ to $v$ and a directed edge from $v$ to each vertex of $L_1\cup L_2$. There is a directed edge from each vertex of $L_{k}\cup L_{k+1}$ to $u$ and a directed edge from $u$ to each vertex of $L'_{4}$. See Figure \ref{fig:fixed_edges}.

\begin{figure}[h]
  \centering
  \includegraphics[width=\linewidth]{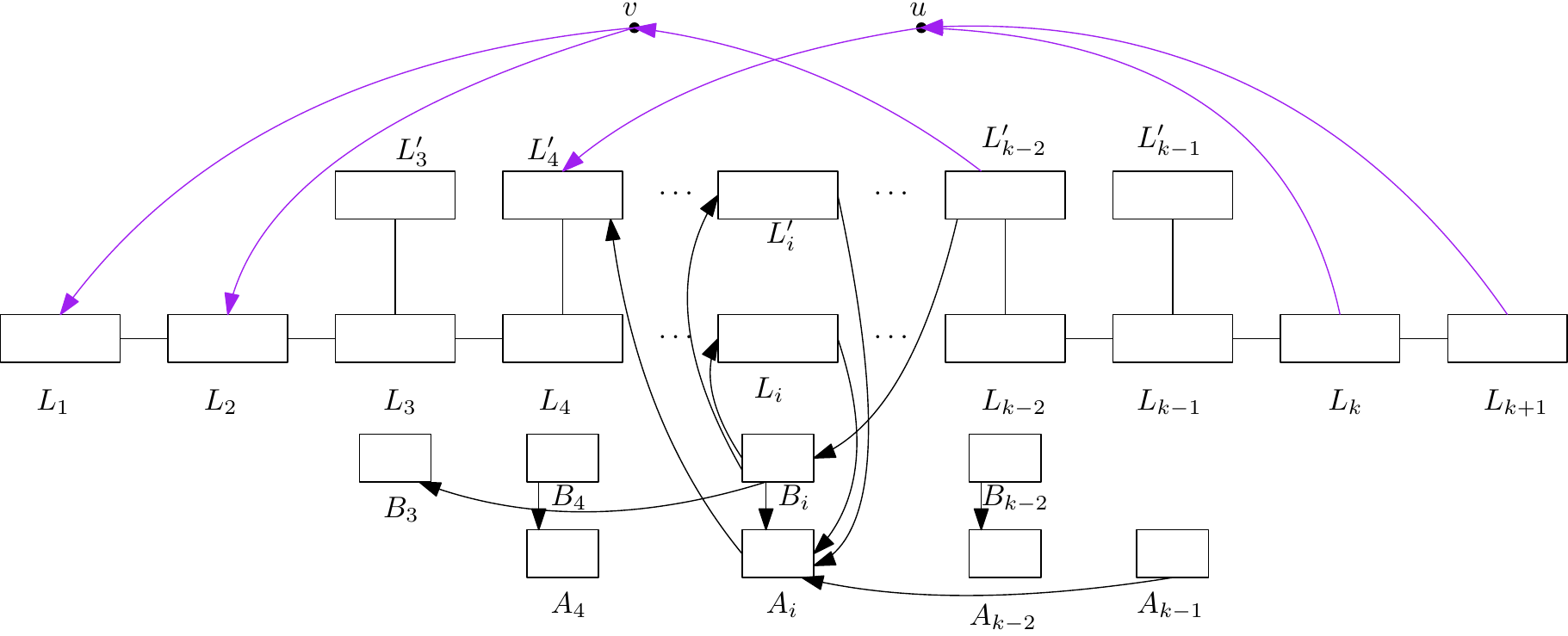}
  \caption{Vertex set of $G$ and its directed edges. The purple edges show fixed edges and they are attached to all nodes in a set they are pointing to/from.}
  \label{fig:fixed_edges}
\end{figure}

This finishes the definition of the graph $G$.

\begin{figure}
    \centering
    \includegraphics[width=0.65\linewidth]{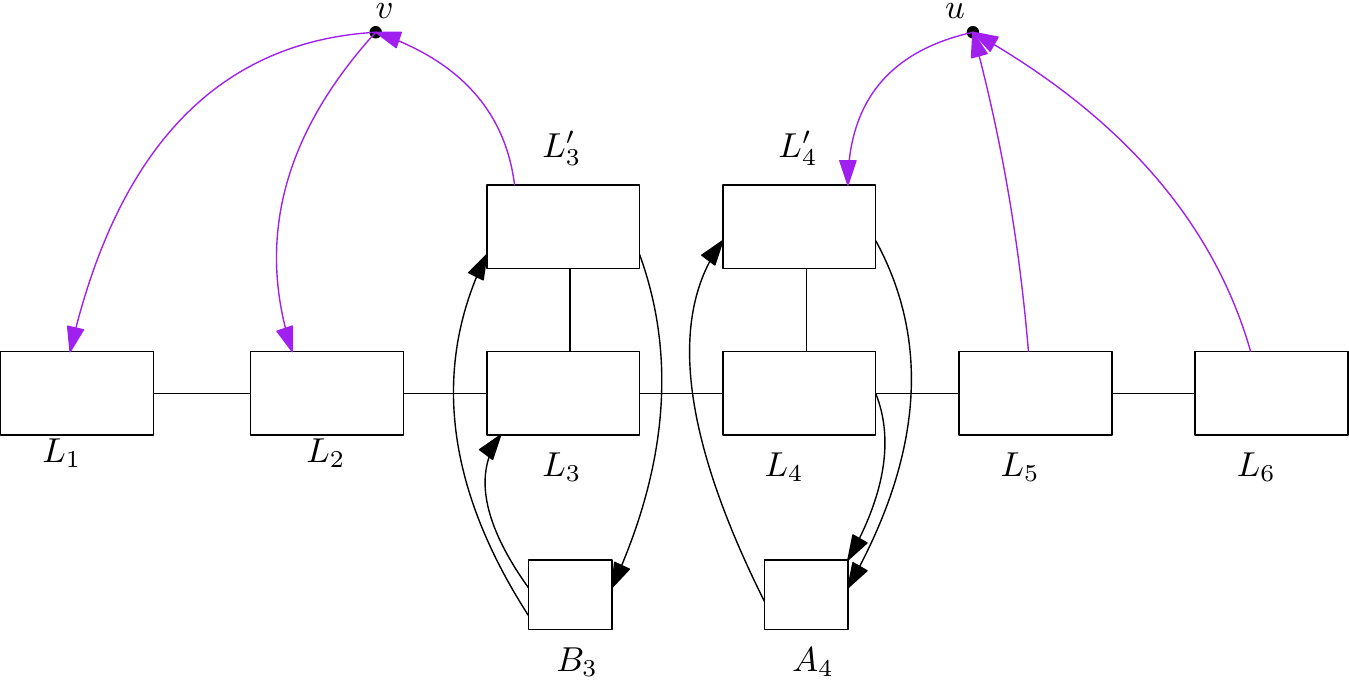}
    \caption{The construction when $k=5$.}
    \label{fig:k5}
\end{figure}

\paragraph{Number of edges:}  Coordinate-change edges and vector-change edges are only incident to vertices in $L'\cup L_2\cup \dots\cup L_k$, of which there are $O(n^{k-2}d^{k-1})$. Each such vertex has at most $d^{k-1}$ incident coordinate-change edges, since this is the total number of possible coordinate arrays. Each such vertex has at most $n$ incident vector-change edges, since each vector-change edge only changes one vector. Thus, there are $O(n^{k-1}d^{2k-2})$ coordinate-change and vector-change edges.

Swap edges are only incident to vertices in $L$. Each of the $O(n^{k-1})$ vertices in $L_1\cup L_{k+1}$ is incident to at most $d^{k-1}$ swap edges since swap edges from $L_1$ to $L_2$ and from $L_{k+1}$ to $L_k$ change a vector for a coordinate. Each of the $O(n^{k-2}d^{k-1})$ vertices in $L_2\cup\dots\cup L_k$ is incident to at most $n$ swap edges since these edges change a vector or coordinate for a vector. Thus, there are $O(n^{k-1}d^{k-1})$ swap edges.

Each vertex is incident to $O(k)$ $a$-type or $b$-type back edges since each vertex has at most one edge to and from each $A_i$ and $B_i$. The number of $ba$-type back edges between each pair $B_i,A_i$ is $O(n^{k-2})$, since $B_i$ has at most $n^{k-i-2}$ nodes and $A_i$ has at most $n^{i}$ nodes. 

Finally, each vertex is incident to at most two fixed edges.

Thus, we have shown that the total number of edges is $O(n^{k-1}d^{2k-2})$.

 \section{NO instance of $k$-OV implies diameter at most $k$}
 \label{sec:k>5no}
In a NO instance, for every set $F$ of at most $k$ vectors, there exists a coordinate that is 1 for every vector in $F$. Given a set $F$ of at most $k$ vectors, we let $C(F)$ denote a coordinate that is 1 for every vector in $F$. For any vertex $\alpha$, let $C(\alpha)$ be a coordinate that is 1 for every vector in $\alpha$.
   
 \subsection{Fixed paths}\label{sec:fixed}
 First, we will calculate the distance between pairs of vertices whose distance does not depend on the answer to the $k$-OV instance. Note that regardless of the answer to the $k$-OV instance, we can assume that every set $F$ of at most $k-1$ vectors are not orthogonal. That is, $C(F)$ and $C(\alpha)$ are well-defined for any $|F|\leq k-1$ and any $\alpha$.

\begin{claim}
\label{claim:fixedpaths}
For each vertex $\alpha\in V$:
 \begin{enumerate}
\item There exists a vertex $\beta_1\in L'_{k-2}$ so that $d(\alpha,\beta_1)\leq k-2$,
\item There exists a vertex $\beta_2\in L'_4$ so that $d(\beta_2,\alpha)\leq k-2$,
\item There exists a vertex $\beta_3\in L_{k}$ so that $d(\alpha,\beta_3)\leq k-1$, and
\item There exists a vertex $\beta_4\in L_2$ so that $d(\beta_4,\alpha)\leq k-1$.
\end{enumerate}
\end{claim} 

\begin{proof}
 We prove 1 and 3. Then, 2 and 4 follow due to symmetry. Starting from $\alpha$, we can proceed towards the appropriate layer ($L'_{k-2}$ or $L_{k}$) as follows. First, by taking at most two edges we can go to a vertex in $L_2\cup\dots\cup L_{k}$ as follows. If $\alpha\in A\cup B$, the first edge is a back edge to a vertex in $L'_j$ for some $j$; note that such a vertex exists because the new vectors can all be the all 1s vector, and the coordinate array can be $k-1$ copies of $C(\alpha)$. The second edge is a coordinate-change edge to $L_j$ that does not actually change the coordinate. If $\alpha$ is in an $L'_i$ set, we only take one edge which is a coordinate-change edge, where the new coordinate array is $k-1$ copies of $C(\alpha)$. If $\alpha=u$, we take a fixed edge to $L'_4$ and then an edge to $L_4$. If $\alpha=v$, we take a fixed edge to $L_2$. If $\alpha\in L_1\cup L_{k+1}$, then we take an edge to $L_2$ or $L_k$ (respectively) by adding the coordinate array that is $k-1$ copies of $C(\alpha)$. 

 Then we proceed to the $L_{k-2}$ by taking swap edges that change some vector to the all 1s vector. Then we take a vector-change edge to $L'_{k-2}$ or two swap edges to $L_{k}$. The vector-change edge need not actually change any vectors. It is straightforward to see that these paths are of the appropriate lengths.
\end{proof}

 Due to the fixed edges in the graph, the consequences of Claim~\ref{claim:fixedpaths} are respectively that
 \begin{enumerate}
     \item For each vertex $\alpha\in V$ and for any vertex $\beta\in L_1\cup L_2\cup\{v\}$, $d(\alpha,\beta)\leq k$.
     \item For each vertex $\alpha\in V$ and for any vertex $\beta\in L_{k}\cup L_{k+1}\cup\{u\}$, $d(\beta,\alpha)\leq k$.
     \item For each vertex $\alpha\in V$, $d(\alpha,u)\leq k$.
     \item For each vertex $\alpha\in V$, $d(v, \alpha)\leq k$.
 \end{enumerate}

 \subsection{Variable paths }
The distances that we did not bound in Section~\ref{sec:fixed} are those from a vertex $\alpha$ to a vertex $\beta$ in the following cases. Cases 1 through 3 demonstrate paths with both endpoints in $L\cup L'$. %$L_1\cup \ldots \cup L_{k-1}\cup L'_3\cup \ldots\cup L'_{k-1}$. 
Cases 4 through 7 demonstrate paths with one endpoint in $A$ or $B$. 
 \begin{enumerate}
     \item $\alpha\in L_1\cup L_2$ and $\beta\in L'\cup L\setminus (L_1\cup L_2)$.  %V\setminus (\{u,v\}\cup L_1\cup L_2\cup A\cup B)$. 
     \item $\alpha\in L'\cup L\setminus (L_k\cup L_{k+1})$ %V\setminus (\{u,v\}\cup L_k\cup L_{k+1}\cup A\cup B)$
     and $\beta\in L_{k}\cup L_{k+1}$. 
     \item $\alpha$, $\beta \in L_3\cup\dots\cup L_{k-1}\cup L'_3\cup\dots\cup L'_{k-1}$.
     \item $\alpha\in V\setminus\{u,v\}$ and $\beta\in B$.
     \item $\alpha\in B$ and $\beta\in V\setminus\{u,v\}$. 
     \item $\alpha\in V\setminus\{u,v\}$ and $\beta\in A$.
     \item $\alpha\in A$ and $\beta\in V\setminus\{u,v\}$. 
     
 \end{enumerate}

 Cases 1 and 2 are completely symmetric, as are cases 4 and 7 as well as cases 5 and 6. Hence we only analyze cases 1, 3, 4, and 5.

 \paragraph{Case 1: $\alpha\in L_1\cup L_2$ and $\beta\in L'\cup L\setminus (L_1\cup L_2)$.} If $\alpha\in L_1$ let $\alpha=(a_1,\dots,a_{k-1})$ and if $\alpha\in L_2$ let $\alpha=(a_1,\dots,a_{k-2},x_\alpha)$ in which case we let $a_{k-1}$ be the all 1s vector. We condition on which set $\beta$ is in.
 
 \subparagraph{Case 1a: $\beta\in L_k\cup L_{k+1}$.}  If $\beta\in L_{k+1}$ let $\beta=(b_2,\dots,b_{k})$, and if $\beta\in L_{k}$ let $\beta=(b_3,\dots,b_k,x_\beta)$ in which case we let $b_2$ be the all 1s vector. 

We will use the coordinate array $x=(x_1,\dots,x_{k-1})$ defined as follows: for all $1\leq \ell\leq k-1$, $x_{k-\ell}=C(a_1,\dots,a_\ell,b_{\ell+1},\dots,b_k)$. See Table \ref{tab:case1}(a).

\begin{table}[!htb]
    
    \begin{subtable}{.5\linewidth}
      \caption{Cases 1a and 4}
      \centering
         \begin{tabular}{|c|c|c|c|c|c|} 
 \hline
&$x_1$ & $\ldots$ & $x_{\ell}$ & $\ldots$ & $x_{k-1}$\\ \hline
$a_1$&$1$ & $\ldots$ & $1$ &$\ldots$ & $1$ \\ \hline
$\ldots$& $\ldots$&$\ldots$ &$\ldots$ &$\ldots$ & \\ \hline
$a_{k-\ell}$& $1$&$\ldots$  & $1$ & &  \\ \hline
$\ldots$&$\ldots$ &$\ldots$& & &  \\ \hline
$a_{k-1}$ & $1$& & & & \\ \hline
$b_2$ &  & & & &$1$ \\ \hline
$\ldots$& & & &$\ldots$ &$\ldots$ \\ \hline
$b_{k-\ell+1}$ & &&  $1$ &$\ldots$  & $1$ \\ \hline
$\ldots$&&$\ldots$ &$\ldots$ &$\ldots$ &  $\ldots$\\ \hline
$b_k$&$1$ & $\ldots$ & $1$ &$\ldots$ & $1$ \\ \hline
\end{tabular}

    \end{subtable}%
    \begin{subtable}{.5\linewidth}
      \centering
        \caption{Cases 1b and 1c}
         \begin{tabular}{|c|c|c|c|c|c|c|} 
 \hline
&$x_1$ & $\ldots$ & $x_{\ell}$ & $\ldots$ & $x_{k-2}$&$x_{k-1}$\\ \hline
$a_1$&$1$ & $\ldots$ & $1$ &$\ldots$ & $1$ &$1$\\ \hline
$a_2$&$1$ & $\ldots$ & $1$ &$\ldots$ & $1$ &$1$\\ \hline
$\ldots$& $\ldots$&$\ldots$ &$\ldots$ &$\ldots$ & & \\ \hline
$a_{k-\ell}$&$1$&$\ldots$ & $1$ &  & &  \\ \hline
$\ldots$& $\ldots$&$\ldots$ & & & &\\ \hline
$a_{k-1}$ &$1$ & & & & &\\ \hline
$c_1$ &  & & & & $1$& $1$\\ \hline
$\ldots$& & & & $\ldots$ &$\ldots$&$\ldots$\\ \hline
$c_{k-\ell-1}$ &  & & $1$&$\ldots$ & $1$ & $1$ \\ \hline
$\ldots$& &$\ldots$ &$\ldots$ &$\ldots$ & $\ldots$& $\ldots$\\ \hline
$c_{k-2}$&$1$ & $\ldots$ & $1$ &$\ldots$ & $1$& $1$ \\ \hline
\end{tabular}
    \end{subtable} 
    \caption{Coordinate arrays used in the $\alpha\beta$ path. 
    }
    \label{tab:case1}
\end{table}

We now describe a path of length $k$ from $\alpha$ to $\beta$. All of the internal vertices along the path use the coordinate $x$, and the existence of these vertices follows from the definition of $x$.

We begin by taking an edge from $\alpha$ to $(a_1,\dots,a_{k-2},x)\in L_2$ either by taking a swap edge or a coordinate-change edge (depending on whether $\alpha$ is in $L_1$ or $L_2$). We then use swap edges to go from $L_2$ to $L_{k}$, where to go from $L_{r}$ to $L_{r+1}$ we change the vector $a_{k-r}$ to the vector $b_{k-r+2}$. After traversing these edges, we end at $(b_3,\dots,b_k,x)\in L_{k}$. Finally, we take an edge to $\beta$ (which is a swap edge or a coordinate-change edge depending on whether $\beta$ is in $L_{k+1}$ or $L_{k}$). This path is shown with black (solid and dashed lines) and blue edges in Figure \ref{fig:case1}.
\subparagraph{Case 1b: $\beta\in L_{k-1}\cup L'_{k-1}$.} Let $\beta=(a'_1,b_4,\ldots,b_k,x_{\beta})=(c_1,\dots,c_{k-2},x_\beta)\in L_{k-1}\cup L'_{k-1}$. We will use the coordinate array $x=(x_1,\dots,x_{k-1})$ defined as follows: for all $1\leq \ell\leq k-2$, $x_\ell=C(a_1,\dots,a_{k-\ell},c_{k-\ell-1},\dots,c_{k-2})$, and we set $x_{k-1}=x_{k-2}$. See Table \ref{tab:case1}(b). We begin by taking an edge from $\alpha$ to $(a_1,\dots,a_{k-2},x)\in L_2$ either by taking a swap edge or a coordinate-change edge (depending on whether $\alpha$ is in $L_1$ or $L_2$). We then use swap edges to go from $L_2$ to $L_{k-1}$, where to go from $L_{r}$ to $L_{r+1}$ we change the vector $a_{k-r}$ to the vector $b_{k-r+2}$. So we are at $(a_1, b_4,\ldots,b_k,x)\in L_{k-1}$. We then take a vector-change edge to $(a'_1,b_4,\ldots,b_k,x)\in L'_{k-1}$, and then take a coordinate-change edge to $\beta$. This path is shown with black (solid and dashed lines) and red edges in Figure \ref{fig:case1}.

\subparagraph{Case 1c: $\beta \in L_3\cup\dots\cup L_{k-2}\cup L'_3\cup\dots\cup L'_{k-2}$.} Suppose $\beta\in L_i\cup L'_i$ and let $\beta=$\\ $(a'_1,\dots,a'_{k-i},b_{k-i+3},\dots,b_k,x_\beta)=(c_1,\dots,c_{k-2},x_\beta)$.

We will use the coordinate array $x=(x_1,\dots,x_{k-1})$ defined as follows: for all $1\leq \ell\leq k-2$, $x_\ell=C(a_1,\dots,a_{k-\ell},c_{k-\ell-1},\dots,c_{k-2})$, and we set $x_{k-1}=x_{k-2}$. See Table \ref{tab:case1}(b).

%We now describe a path of length $k$ from $\alpha$ to $\beta$. All of the internal vertices along the path use the coordinate $x$, and the existence of these vertices follows from the definition of $x$.

We begin by taking an edge from $\alpha$ to $(a_1,\dots,a_{k-2},x)\in L_2$ either by taking a swap edge or a coordinate-change edge (depending on whether $\alpha$ is in $L_1$ or $L_2$). We then use swap edges to go from $L_2$ to $L_{k-2}$, where to go from $L_{r}$ to $L_{r+1}$ we change the vector $a_{k-r}$ to the vector $b_{k-r+2}$, where $b_{k-r+2}$ has been defined as part of $\beta$ if $r\leq i-1$, and otherwise we define $b_{k-r+2}$ as the all 1s vector. After traversing these edges, we end at $(a_1,a_2,b_5,\dots,b_k,x)\in L_{k-2}$. 

Then, we take a coordinate-change edge to arrive at $(a_1,a_2,b_5,\dots,b_k,x)\in L'_{k-2}$ (without actually changing any coordinates). 
This vertex exists by option (2) of the specification of vertices in $L'_{k-2}$. So far, the path is of length $k-2$. 

Then, we use a $b$-type back edge to go to $\gamma=(b_{k-i+3},\ldots,b_k)\in B_i$, and then use another $b$-type back edge to go from $\gamma$ to $\beta$.
The full path is of length $k$. This path is shown with black (solid and dashed lines) and purple edges in Figure \ref{fig:case1}.
\begin{figure}
    \centering
    \includegraphics[width=0.8\linewidth]{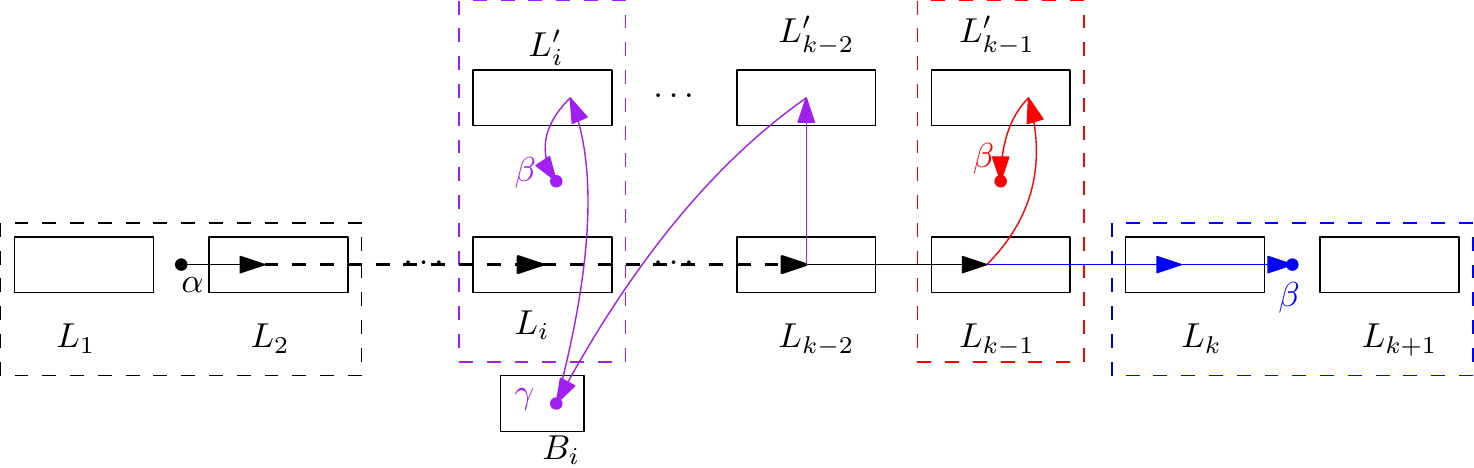}
    \caption{$\alpha\beta$ path in Case 1. Different cases of $\beta$ are shown with different colors. The path between $L_2$ and $ L_{k-2}$ is shown with black dashed lines. Solid lines indicate edges. A node in a dashed box containing two sets means that the node is in either set.}
    \label{fig:case1}
\end{figure}

\paragraph{Case 3: $\alpha$, $\beta \in L_3\cup\dots\cup L_{k-1}\cup L'_3\cup\dots\cup L'_{k-1}$.} %\nicole{Renumbered this case from 3 to 2, and also the following cases.} 
Suppose $\alpha\in L_i\cup L'_i$ and $\beta\in L_j\cup L'_j$. Let $\alpha=(a_1,\dots, a_{k-i},b_{k-i+3},\dots ,b_{k},x_\alpha)=(c_1,\dots,c_{k-2},x_\alpha)$ and let $\beta=(a'_1,\dots, a'_{k-j},b'_{k-j+3},\dots ,b'_{k},x_\beta)=(c'_1,\dots,c'_{k-2},x_\beta)$.

 We will use the coordinate array $x=(x_1,\dots,x_{k-1})$ defined as follows: for all $2\leq\ell\leq k-2$, $x_\ell=C(c_1,\dots,c_{k-\ell},c'_{k-\ell-1},\dots,c'_{k-2})$, and we set $x_1=x_2$ and $x_{k-1}=x_{k-2}$. See Table \ref{tab:case3}.  
 
 \begin{table}[!htb]
      \centering
 \begin{tabular}{|c|c|c|c|c|c|c|c|} 
 \hline
&$x_1$&$x_2$ & $\ldots$ & $x_{\ell}$ & $\ldots$ & $x_{k-2}$&$x_{k-1}$\\ \hline
$c_1$&$1$ &$1$& $\ldots$ & $1$ &$\ldots$ & $1$ &$1$\\ \hline
$c_2$&$1$&$1$ & $\ldots$ & $1$ &$\ldots$ & $1$ &$1$\\ \hline
$\ldots$& $\ldots$& $\ldots$&$\ldots$ &$\ldots$ &$\ldots$ & & \\ \hline
$c_{k-\ell}$&$1$&$1$&$\ldots$ & $1$ &  & &  \\ \hline
$\ldots$& $\ldots$&$\ldots$&$\ldots$ & & & &\\ \hline
$c_{k-2}$ &$1$ &$1$ & & & & &\\ \hline
$c'_1$ & & & & & & $1$& $1$\\ \hline
$\ldots$&& & & & $\ldots$ &$\ldots$&$\ldots$\\ \hline
$c'_{k-\ell-1}$ & & & & $1$&$\ldots$ & $1$ & $1$ \\ \hline
$\ldots$& &&$\ldots$ &$\ldots$ &$\ldots$ & $\ldots$& $\ldots$\\ \hline
$c'_{k-3}$&$1$&$1$ & $\ldots$ & $1$ &$\ldots$ & $1$& $1$ \\ \hline
$c'_{k-2}$&$1$&$1$ & $\ldots$ & $1$ &$\ldots$ & $1$& $1$ \\ \hline
\end{tabular}
\caption{Case 2 coordinate array used in the $\alpha\beta$ path.}
\label{tab:case3}
\end{table}
 
 We now describe a path of length $k$ from $\alpha$ to $\beta$. Follow Figure \ref{fig:case3} for an illustration of the path. All of the internal vertices along the path use the coordinate $x$, and the existence of these vertices follows from the definition of $x$. 
 
 We first show a path of length $3$ from $\alpha$ to $\gamma=(a_1,\ldots,a_{k-i},a_{k-i+1},\ldots,a_{k-4},b'_{k-1},b'_k,x)\in L_4$, where $a_{\ell}$ is the all 1 vector for all $\ell>k-i$ and if $j=3$, $b'_{k-1}$ is the all 1 vector (otherwise $b'_{k-1}$ is defined for $\beta$). 
 
 If $i>3$, using an $a$-type back edge we go from $\alpha$ to $(a_1,\ldots,a_{k-i})\in A_i$, and using another $a$-type back edge we go to $(a_1,\ldots,a_{k-i},{a}_{k-i+1},\ldots,{a}_{k-4},b'_{k-1},b'_k,x)\in L'_4$. Now using a coordinate-change edge we go to $\gamma\in L_4$, without actually changing any coordinates. This corresponds to the green path in Figure \ref{fig:case3}.
 
 If $i=3$, we first take a coordinate-change edge to $(a_1,\ldots,a_{k-3},b_k,x)\in L'_3$, then take a vector-change edge to change $b_k$ to $b'_k$ and arrive at $(a_1,\ldots,a_{k-3},b'_k,x)\in L_3$, and then use a swap edge to go to $\gamma\in L_4$. This corresponds to the purple path in Figure \ref{fig:case3}.

Next, we show a path of length $3$ from $\gamma' = (a_1,a_2,b'_5,\ldots,b'_k,x)\in L_{k-4}$ to $\beta$, where $b'_{\ell}$ is the all 1s vector for $\ell<k-j+3$ and $a_2$ is the all 1s vector if $j=k-1$. 
 Then we show that we can go from $\gamma$ to $\gamma'$ in $k-6$ when $k\geq 6$. We handle $k=5$ in the last paragraph of Case 2. For now, we assume that $k>5$.
 
 Now, we use swap edges to go from $\gamma\in L_4$ to $L_{k-2}$ where for all $4\le r<k-2$ to go from $L_{r}$ to $L_{r+1}$ we change the vector $a_{k-r}$ to the vector $b'_{k-r+2}$, where $b'_{k-r+2}$ has already been defined as part of $\beta$ if $r\leq j-1$, and otherwise we define $b'_{k-r+2}$ as the all 1s vector. This path of swap edges is of length $k-6$, so the path is thus far of length $k-3$. After traversing these edges, we end at $\gamma'=(a_1,a_2,b'_5,\ldots,b'_k,x)\in L_{k-2}$. The $\gamma\gamma'$ path is shown with black dashed lines in Figure \ref{fig:case3}. If $j<k-1$, then using a coordinate-change edge we go to  $(a_1,a_2,b'_5,\ldots,b'_k,x)\in L'_{k-2}$ without actually changing any coordinates. Then we take a $b$-type back edge to go to $(b'_{k-j+3},\ldots,b'_k)\in B_j$, and using another $b$-type back edge we go to $\beta$. This corresponds to the orange path in Figure \ref{fig:case3}.

If $j=k-1$, we use a swap edge to go to $(a_1,b'_4,\ldots,b'_k,x)\in L_{k-1}$, then use a vector-change edge to change $a_1$ to $a'_1$ and arrive at $L'_{k-1}$, and then use a coordinate-change edge to $\beta$. This corresponds to the red path in Figure \ref{fig:case3}. The total length of the path is hence $k$. 

This concludes Case 2 when $k>5$. Now suppose that $k=5$. We have already shown that there is a path of length 3 from $\alpha$ to $\gamma=(a_1,b'_4,b'_5,x)\in L_4$. If $j=4$, we get to $\beta$ from $\gamma$ using a vector-change edge to $(a'_1,b'_4,b'_5,x)\in L'_4$, and then coordinate-change edge. So suppose that $j=3$. Then there is a path of length $3$ from $\gamma''=(a_1,a_2,b_5,x)\in L_3$ to $\beta$ as follows: take a vector-change edge to $(a_1,a_2,b'_5,x)\in L'_3$, then a back edge to $(b'_5)\in B_3$, and then another back edge to $\beta$. Now to go from $\alpha$ to $\gamma''$ in at most $2$, we do the following: If $i=3$, we take a coordinate-change edge from $\alpha$ to $\gamma''$. If $i=4$, we first take a coordinate-change edge to $(a_1,b_4,b_5,x)\in L_4$, and then a swap edge to $\gamma''$.
\begin{figure}[h]
  \centering
  \includegraphics[width=0.8\linewidth]{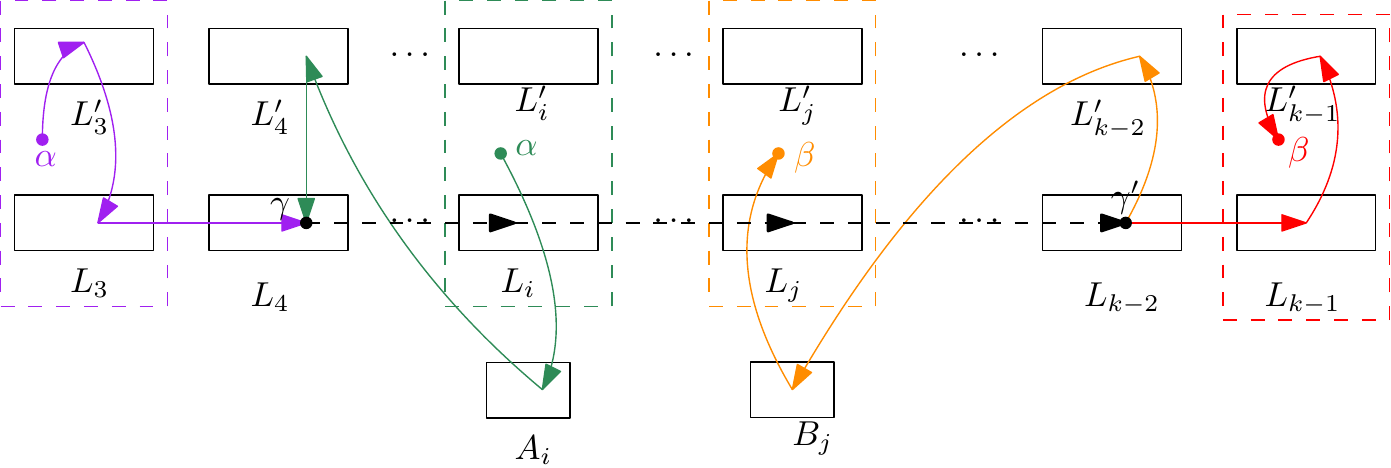}
  \caption{$\alpha\beta$ path in Case 3. Different cases of $\alpha$ and $\beta$ are shown with different colors. The path between $\gamma\in L_4$ and $\gamma'\in L_{k-2}$ is shown with black dashed lines. Solid lines indicate edges. A node in a dashed box containing two sets means that the node is in either set.}
  \label{fig:case3}
\end{figure}

\paragraph{Case 4: $\alpha\in V\setminus\{u,v\}$ and $\beta\in B$:} 
For the majority of this case, we assume that $k>5$, and then at the end of this case we include a paragraph to handle $k=5$.
Let $\beta\in B_j$ and $\beta=(b_{k-j+3},\ldots,b_k)$.

 An intermediate node on the path from $\alpha$ to $\beta$ will be some $\gamma\in L_4$ of the form $(a_1,\ldots,a_{k-4},b_{k-1},b_k,x)$, where $b_{k}$ has already been defined, $b_{k-1}$ has been defined unless $j=3$ in which case $b_{k-1}$ is the all 1s vector, each $a_1,\ldots,a_{k-4}$ is some vector from the appropriate set $S_1,\ldots, S_{k-4}$, and $x$ is a coordinate array satisfying the following conditions: For $\ell=4,\dots,k-4$, $x_{\ell}=C(a_1,\ldots,a_{k-\ell},b_{k-\ell+1},\ldots,b_k)$, where for $r<k-j+3$, $b_r$ is the all 1s vector, and for $r>k-4$, $a_r$ is the all 1s vector. 
 See Table~\ref{tab:case1}. Each vertex in the path from $\alpha$ to $\beta$ that we will specify exists due to the definition of $x$.

First, we show that for \emph{all} $\gamma$ of the above form, there is a path of length at most $k-3$ from $\gamma$ to $\beta$. Using swap edges we go from $\gamma$ to $\gamma'=(a_1,a_2,b_{5},\ldots,b_k,x)\in L_{k-2}$, where $b_{s}$ has already been defined for for $s\geq k-j+3$, and $b_s$ is the all 1s vector for $s< k-j+3$. To construct this path from $\gamma$ to $\gamma'$, for all $r=4,\dots,k-2$, to go from $L_{r}$ to $L_{r+1}$ we change the vector $a_{k-r}$ to the vector $b_{k-r+2}$. Then from $\gamma'$ we take an edge to $(a_1,a_2,b_5,\ldots,b_k,x)\in L'_{k-2}$  and then take a $b$-type back edge to $\beta$. The $\gamma\beta$ path is specified with dashed black lines representing subpaths and  black edges in Figure \ref{fig:case4}.

To complete the path from $\alpha$ to $\beta$, we will show a path of length at most $3$ from $\alpha$ to \emph{some} $\gamma$ of the above form. We divide into cases based on where $\alpha$ is. Because $a_1,\dots,a_{k-4}$ are unspecified in the definition of $\gamma$, we have the freedom to specify these vectors in the following cases.

\subparagraph{Case 4a: $\alpha \in L_k\cup L_{k+1}$.} From $\alpha$ we take a fixed edge to $u$ and then from $u$ we take a fixed edge to $(a_1,\ldots,a_{k-4},b_{k-1},b_{k},x)\in L'_4$, where we define $a_i$ for $i=1,\ldots,k-4$ as the all 1s vector.
Using a coordinate-change edge we go to $\gamma=(a_1,\ldots,a_{k-4},b_{k-1},b_{k},x)\in L_4$ without actually changing any coordinates. This path is illustrated in purple in Figure \ref{fig:case4}.

\subparagraph{Case 4b: $\alpha \in L_1\cup L_2$.} Let $\alpha=(a_1,\ldots,a_{k-1})$ if $\alpha\in L_1$ and let $\alpha = (a_1,\ldots,a_{k-2},x_{\alpha})$ if $\alpha \in L_2$ in which case $a_{k-1}$ is the all 1s vector. 
We begin by taking an edge from $\alpha$ to $(a_1,\dots,a_{k-2},x)\in L_2$ either by taking a swap edge or a coordinate-change edge (depending on whether $\alpha$ is in $L_1$ or $L_2$). We then use two swap edges to go from $L_2$ to $L_{4}$, where to go from $L_{r}$ to $L_{r+1}$ we change the vector $a_{k-r}$ to the vector $b_{k-r+2}$. So we arrive at $\gamma=(a_1,\ldots,a_{k-4},b_{k-1},b_{k},x)$. This path is illustrated in green in Figure \ref{fig:case4}.

\subparagraph{Case 4c: $\alpha \in L_3\cup L'_3\cup B_3$.}  Let $\alpha = (a_1,\ldots,a_{k-3},b'_k,x_{\alpha})=(c_1,\ldots,c_{k-2},x_{\alpha})$ if $\alpha\in L_3\cup L'_3$, and let $\alpha = (b'_{k})=(c_{k-2})$ if $\alpha\in B_3$, in which case $a_{\ell}$ is the all 1s vector for $\ell=1,\ldots,k-3$. 
From $\alpha$, we take a back edge or a coordinate-change edge to $(a_1,\ldots,a_{k-3},b'_k,x)\in L'_3$. Then take a vector-change edge to change $b'_k$ to $b_k$ and arrive at $(a_1,\ldots,a_{k-3},b_k,x)\in L_3$. Then using a swap edge we proceed to $\gamma=(a_1,\ldots,a_{k-4},b_{k-1},b_k,x)\in L_4$. This path is illustrated in red in Figure \ref{fig:case4}. 

\subparagraph{Case 4d: $\alpha \in L_i\cup L'_i\cup B_i\cup A_i$ for $i=4,\ldots,k-1$.} First if $\alpha\notin A_i$, we show how to get to a node in $A_i$. Then we show how to go to $\gamma$ from any node in $A_i$ using a path of length $2$.
Suppose that $\alpha=(a_1,\ldots,a_{k-i},b'_{k-i+3},\ldots,b'_k,x_{\alpha})$ if $\alpha\in L_i\cup L'_i$, and suppose that $\alpha=(b'_{k-i+3},\ldots,b'_k)$ if $\alpha\in B_i$, in which case we assume that $a_{\ell}$ is the all 1s vector for $\ell=1,\ldots,k-i$. Take an $a$-type or $ba$-type back edge to $(a_1,\ldots,a_{k-i})\in A_i$. Now we show how to proceed from this vertex to $\gamma$. 

We take an $a$-type back edge to $(a_1,\ldots,a_{k-4},b_{k-1},b_k,x)\in L'_4$, and then take a coordinate-change edge to $\gamma=(a_1,\ldots,a_{k-4},b_{k-1},b_k,x)\in L_4$ without actually changing any coordinates. This path is illustrated in orange in Figure \ref{fig:case4}. 

\begin{figure}[h]
  \centering
  \includegraphics[width=\linewidth]{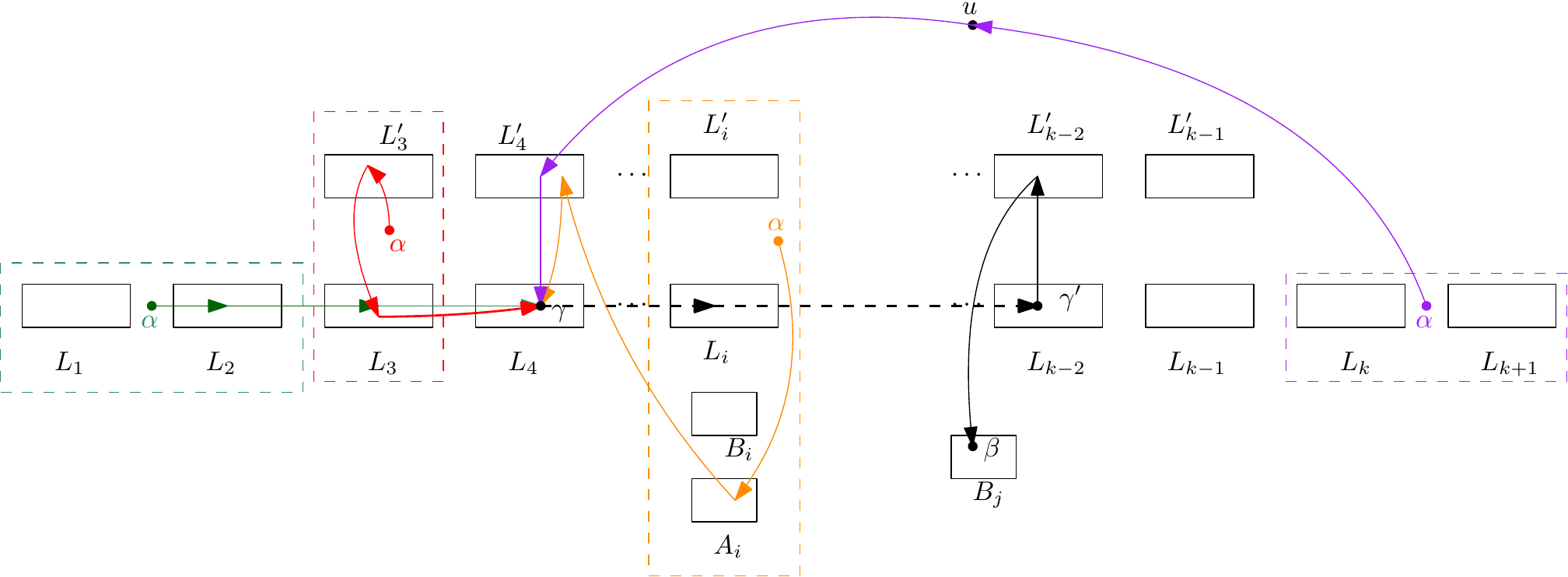}
  \caption{$\alpha\beta$ path in Case 4. Different cases of $\alpha$ are shown with different colors. The path between $\gamma\in L_4$ and $\gamma'\in L_{k-2}$ is shown with black dashed lines. Solid lines indicate edges. A node in a dashed box containing two sets means that the node is in either set.
 }
  \label{fig:case4}
\end{figure}
This completes Case 4 when $k>5$. Now we let $k=5$ and specify a path of length at most 5 from $\alpha$ to $\beta$. Since $k=5$, we have that $B=B_3$ so let $\beta=(b_k)\in B_3$. Regardless of its location in the graph, $\alpha$ has at most 4 vectors in its representation. Let $a_1,a_2,a_3,a_4$ be a set consisting of all of the vectors in the representation of $\alpha$ plus some all 1s vectors if $\alpha$ has fewer than 4 vectors in its representation. Let $x$ be the coordinate array consisting of 4 copies of the coordinate $C(a_1,a_2,a_3,a_4,b_{k})$. From $\alpha$ we can go to $\gamma=(a_1,a_2,b'_5,x)\in L_3$ using a path of length $3$. To do this, we either take a coordinate-change edge or a back edge to a vertex in $L_2\cup \ldots \cup L_4\cup L'$, and then take swap or coordinate-change edges to $\gamma$. Then from $\gamma$, we take a vector-change edge to $(a_1,a_2,b_5,x)\in L'_3$, and then a back edge to $\beta$.

\paragraph{Case 5:} $\alpha\in B$ and $\beta\in V\setminus \{u,v\}$. Suppose that $\alpha \in B_i$. First suppose that $i>3$. We know that by Claim \ref{claim:fixedpaths}, there is a $\beta'\in L'_4$ such that $d(\beta',\beta)\le k-2$. We show that there is a path of length $2$ from $\alpha$ to $\beta'$. Suppose that $\beta'=(a_1,\ldots,a_{k-4},b'_{k-1},b'_k,x_{\beta'})$. From $\alpha$, take a $ba$-type back edge to $(a_1,\ldots,a_{k-i})\in A_i$, and then take an $a$-type back edge to $\beta'$. 

Now suppose that $i=3$ and $\alpha = (b_k)$. Suppose that $\beta$ is in the $j$th level. We represent $\beta$ as follows: If $\beta\in L_j\cup L'_j$ for $j=2,\ldots,k$, let $\beta=(a_1,\ldots,a_{k-j},b'_{k-j+3},\ldots,b'_k,x_{\beta})$. If $\beta\in L_1$, let $\beta=(a_1,\ldots,a_{k-1})$, and if $\beta\in L_{k+1}$ let $\beta=(b'_2,\ldots,b'_k)$. Finally, if $\beta\in A_j$ for some $j=4,\ldots,k-1$, let $\beta=(a_1,\ldots,a_{k-j})$ and if $\beta\in B_j$ for some $j=3,\ldots,k-2$, let $\beta=(b_{k-j+3},\ldots,b_k)$. Define the coordinate array $x$ such that for all $\ell=1,\ldots,k-1$, $x_{\ell}=C(a_1,\ldots,a_{k-j},b'_{k-j+3},\ldots,b'_k,b_k)$, where for $j=1$, $x_{\ell}=C(a_1,\ldots,a_{k-1},b_k)$ and for $j=k+1$, $x_{\ell}=C(b'_2,\ldots,b'_k,b_k)$. We show a path of length at most $k$ to $\beta$: First take a $b$-type back edge to $(a_1,\ldots,a_{k-3},b_k,x)\in L'_3$, where $a_{\ell}$ is the all 1s vector for $\ell>k-j$. Then take a vector-change edge to change $b_k$ to $b'_k$ and arrive at $(a_1,\ldots,a_{k-3},b'_k,x)\in L_3$. 

If $\beta\in L_j$ for $j=1,2,k,k+1$, take swap edges to $\beta$. This path is of length $2+|j-3|\le k$, and it is shown in Figure \ref{fig:case5}: For $j=1,2$ it is shown with black and green and for $j=k,k+1$ it is shown with black (solid and dashed) and blue.

If $\beta\in L_j\cup L'_j$ for some $j=3,\dots,k-1$, take swap edges to get to 
$(a_1,\ldots,a_{k-j},b'_{k-j+3},\ldots,b_{k-j},x)\in L_j$. So far the path is of length $2+|j-3|\le k-2$. If $\beta\in L'_j$, take one coordinate-change edge to $\beta$. Note that the $\alpha\beta$ path in this case is of length $k-1$. If $\beta\in L_j$, take one coordinate-change edge to the copy of $\beta$ in $L'_j$, and then another edge to $\beta
\in L_j$. These paths are shown with black (dashed and solid) and purple in Figure \ref{fig:case5}.

If $\beta\in A\cup B$, then there is a $\beta'\in L'_j$ for some $j$, such that $d(\beta',\beta)=1$. Above, we showed how to get to any $\beta'\in L'_j$ with a path of length at most $k-1$. For $\beta\in A_j$ the path is shown with black and orange in Figure \ref{fig:case5}, and for $\beta\in B_j$ it is shown with black and red.

\begin{figure}[h]
  \centering
  \includegraphics[width=\linewidth]{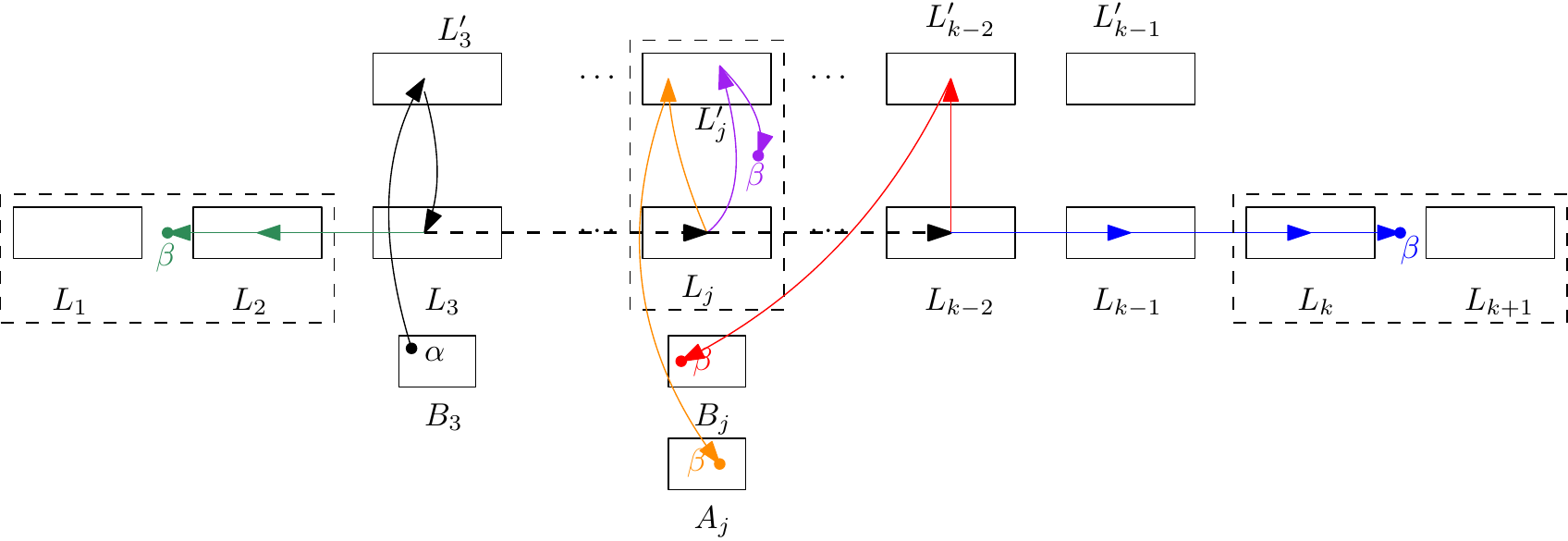}
  \caption{$\alpha\beta$ paths in Case 5 when $\alpha\in B_3$.
Different cases of $\beta$ are shown with different colors. The path between $L_3$ and $L_{k-2}$ is shown with black dashed lines. Solid lines indicate edges. A node in a dashed box containing two sets means that the node is in either sets.}
  \label{fig:case5}
\end{figure}

\section{YES instance of $k$-OV implies diameter at least $2k-1$}
\label{sec:k>5yes}

Let $a_1,a_2,\dots, a_{k}\in S$ be orthogonal. Let $\alpha=(a_1,\dots,a_{k-1})\in L_1$ and let $\beta=(a_2,\dots,a_{k})\in L_{k+1}$. We claim that $d(\alpha,\beta)\geq 2k-1$. Let $P$ be a shortest path from $\alpha$ to $\beta$.

Recall that \emph{level $i$} is defined as $L_i\cup L'_i\cup A_i\cup B_i$. Recall that a \emph{layer} refers to an individual set $L_i$ or $L'_i$. 

We begin with two observations.

\begin{observation}\label{obs:forward}
The only edges that go from a vertex in some level $i$ to a vertex in a level $j>i$ are swap edges between $L_i$ and $L_{i+1}$.
\end{observation}

The next observation follows directly from Observation~\ref{obs:forward}.

\begin{observation}\label{obs:layer}
For all $i<j$, any path from a vertex in level $i$ to a vertex in level $j$ uses a vertex in every layer $L_i,\dots,L_j$. 
\end{observation}

We claim that if $P$ contains either $u$ or $v$ (or both) then the length of $P$ is at least $2k-1$. If $P$ contains $v$ then $P$ must go from $L_1$ to $L_{k-2}$ to $L'_{k-2}$ to $v$ to $L_2$ to $L_{k+1}$, which costs at least $2k-1$ by Observation~\ref{obs:layer}. The argument is symmetric for $u$: If $P$ contains $u$ then $P$ must go from $L_1$ to $L_k$ to $u$ to $L'_4$ to $L_4$ to $L_{k+1}$, which costs at least $2k-1$ by Observation~\ref{obs:layer}. From now on we assume that $P$ does not contain $u$ or $v$.

Next, we claim that if $P$ contains a $ba$-type back edge then the length of $P$ is at least $2k-1$. Since the only edges to $B$ are from $L'_{k-2}$, and the only edges from $A$ are to $L'_4$, if $P$ contains a $ba$-type back edge then $P$ must go from $L_1$ to to $L_{k-2}$ to $L'_{k-2}$ to $B$ to $A$ to $L'_4$ to $L_4$ to $L_{k+1}$. This costs at least $2k-1$ by Observation~\ref{obs:layer}. From now on we assume that $P$ does not contain any $ba$-type back edges. 

We make one more observation, which follows from Observation~\ref{obs:layer} and the following fact: Ignoring $ba$-type edges, all edges from $A$ go to $L'_4$, all edges to $A$ are from a level that is at least 4, all edges to $B$ are from $L'_{k-2}$, and all edges from $B$ are to a level that is at most $k-2$.

\begin{observation}\label{obs:ab}
If $P$ visits $A$, then $P$ visits $L_{4}$ both before and after visiting $A$. If $P$ visits $B$, then $P$ visits $L_{k-2}$ both before and after visiting $B$. 
\end{observation}

\begin{lemma}\label{lem:first}
Fix $i=2,\dots,k$ and let $\gamma_1$ and $\gamma_2$ be the first and last vertices in $L_i$ that $P$ visits, respectively. Let $P_1$ be the subpath of $P$ from $\alpha$ to $\gamma_1$ and let $P_2$ be the subpath of $P$ from $\gamma_2$ to $\beta$. If $P_1$ does not visit $B$, then $\gamma_1$ has the prefix $(a_1,\dots,a_{k-i})$, and if $P_2$ does not visit $A$, then $\gamma_2$ has the suffix $(a_{k-i+3},\dots,a_k,x)$ for some coordinate array $x$.
\end{lemma}

\begin{proof}
Suppose that $P_1$ does not visit $B$. By Observation~\ref{obs:layer},  all vertices in $P_1$ except for $\gamma_1$ are in levels below $i$.
Then by construction, all vertices on $P_1$ contain a prefix of the form $(a'_1,\dots,a'_{k-i})$ where $a'_i\in S_i$. Since $P_1$ does not visit $B$, $u$, or $v$, all edges on $P_1$ do not change any of the vectors $(a'_1,\dots,a'_{k-i})$. Thus, these vectors are the same for $\gamma_1$ and $\alpha$; that is, $\gamma_1$ has the prefix $(a_1,\dots,a_{k-i})$.

Now suppose $P_2$ does not visit $A$. By Observation~\ref{obs:layer}, all vertices in $P_2$ except for $\gamma_2$ are in levels above $i$. 
Then by construction, all vertices on $P_2$ contain a suffix of the form $(a'_{k-i+3},\dots, a'_k)$ where $a'_i\in S_i$. Since $P_2$ does not visit $A$, $u$, or $v$, all edges on $P_2$ do not change any of the vectors $(a'_{k-i+3},\dots, a'_k)$. Thus, these vectors are the same for $\gamma_2$ and $\beta$; that is, $\gamma_2$ ends with $(a_{k-i+3},\dots, a_k, x)$ for some $x$.
\end{proof}

To analyze the length of $P$, we will condition on which $L_i$ are in a \emph{loop}, which is defined as follows.

\begin{definition}[loop]
For all $i$, we say that $L_i$ \emph{is in a loop} if $P$ visits $L_i$ at least twice. 
\end{definition}

An outline of the remainder of the proof is as follows. First we will define the notion of a set $L_i$ in a loop that \emph{covers} a set $L_j$ not in a loop. Then we will prove that every $L_i$ is either in a loop or covered. Then, we will partition $L$ where each piece of the partition is composed of a set of consecutive layers that are all in loops as well as the layers covered by these layers. Then we will compute the length of the subpath of $P$ in each of these pieces. Finally, we will take the sum over all of these subpaths and the edges between them.

\begin{definition}[cover]
For all $i=3,\dots,k-1$, we say that $L_i$ \emph{covers} $L_{i-1}$ if the following conditions are satisfied: \begin{enumerate}
    \item $L_{i-1}$ is not in a loop,
    \item $L_i$ is in a loop, and
    \item If $\gamma_1= (a'_1,\dots, a'_{k-i},a'_{k-i+3},\dots ,a'_{k},x_{\gamma_1})$ and $\gamma_2=(a''_1,\dots, a''_{k-i},a''_{k-i+3},\dots ,a''_{k},x_{\gamma_2})$ are the first and last vertices from $L_i$ that $P$ visits, then $a'_{k-i+3}\not=a''_{k-i+3}$ and $x_{\gamma_1}\not=x_{\gamma_2}$.
\end{enumerate}
Symmetrically, for all $i=3,\dots,k-1$, we say that $L_i$ \emph{covers} $L_{i+1}$ if the following conditions are satisfied: \begin{enumerate}
    \item $L_{i+1}$ is not in a loop,
    \item $L_i$ is in a loop, and
    \item If $\gamma_1= (a'_1,\dots, a'_{k-i},a'_{k-i+3},\dots ,a'_{k},x_{\gamma_1})$ and $\gamma_2=(a''_1,\dots, a''_{k-i},a''_{k-i+3},\dots ,a''_{k},x_{\gamma_2})$ are the first and last vertices from $L_i$ that $P$ visits, then $a'_{k-i}\not=a''_{k-i}$ and $x_{\gamma_1}\not=x_{\gamma_2}$.
    \end{enumerate}
Additionally, we say that $L_4$ covers both $L_3$ and $L_2$ if $P$ visits $A$. Symmetrically, we say that $L_{k-2}$ covers both $L_{k-1}$ and $L_k$ if $P$ visits $B$.
\end{definition}

\begin{lemma}\label{lem:cover}
For all $i=2,\dots,k$, $L_i$ is either in a loop or covered.
\end{lemma}

\begin{proof}
Suppose for contradiction that there exists $i$ with $2\leq i \leq k$ such that $L_i$ is neither in a loop nor covered. Let $\gamma$ be the single vertex in $L_i$ that $P$ visits. Let $P_1$ be the part of $P$ from $\alpha$ to $\gamma$, and let $P_2$ be the part of $P$ from $\gamma$ to $\beta$. By Observation~\ref{obs:forward}, except for $\gamma$, $P_1$ is entirely contained in levels below $i$, and except for $\gamma$, $P_2$ is entirely contained in levels above $i$.

We claim that $P_1$ does not visit $B$. The only edges to $B$ from another set are from $L'_{k-2}$, so $P_1$ can only visit $B$ if $i$ is either $k-1$ or $k$. In this case $L_i$ is covered, by the final part of the definition of cover. Thus $P_1$ does not visit $B$. 

Similarly, $P_2$ does not visit $A$ because the only edges from $A$ to another set are to $L'_4$, so $P_2$ can only visit $A$ if $i$ is either $3$ or $2$, in which case $L_i$ is covered, by the final part of the definition of cover. Since $P_1$ does not visit $B$ and $P_2$ does not visit $A$, Lemma~\ref{lem:first} implies that $\gamma=(a_1,\dots,a_{k-i},a_{k-i+3},\dots,a_k,x)$ for some coordinate array $x$. 

Now consider the two edges incident to $\gamma$ on the path $P$. Suppose that they are $\gamma_1\gamma$ and $\gamma\gamma_2$. Since $L_i$ is not in a loop, Observation~\ref{obs:layer} implies that  $\gamma_1\gamma$ is a swap edge from $L_{i-1}$ to $L_i$.  Hence, if $i\not=2$, then $\gamma_1=(a_1,\dots,a_{k-i},a'_{k-i+1},a_{k-i+4},\dots,a_k,x)$ for some $a'_{k-i+1}$, and if $i=2$, then $\gamma_1=\alpha$. Symmetrically, Observation~\ref{obs:layer} implies that $\gamma\gamma_2$ is a swap edge from $L_i$ to $L_{i+1}$. Hence, if $i\not=k$, then $\gamma_2=(a_1,\dots,a_{k-i-1},a'_{k-i+2},a_{k-i+3}\dots,a_k,x)$ for some $a'_{k-i+2}$, and if $i=k$, then $\gamma_2=\beta$.

By Observation~\ref{obs:layer}, $\gamma_1$ is the last vertex in $L_{i-1}$ that $P$ visits and $\gamma_2$ is the first vertex in $L_{i+1}$ that $P$ visits.
Suppose $i\not=2$ and let $\gamma'_1$ be the first vertex in $L_{i-1}$ that $P$ visits. By condition 3 of the definition of cover, since $L_{i-1}$ does not cover $L_i$, $\gamma'_1$ either contains $a'_{k-i+1}$ or $x$ in its representation. However, since $P_1$ does not visit $B$, Lemma~\ref{lem:first} implies that $\gamma'_1$ has the prefix $(a_1,\dots,a_{k-i+1})$. Thus, either $a_{k-i+1}=a'_{k-i+1}$ or $\gamma'_1$ contains $x$ in its representation. If $a_{k-i+1}=a'_{k-i+1}$, then $\gamma_1$ has the prefix $(a_1,\dots,a_{k-i+1})$, so for all $j=1,\dots,k-i+1$ we have $a_j[x_{i-1}]=1$. If $\gamma'_1$ contains $x$ in its representation, then since $\gamma'_1$ has the prefix $(a_1,\dots,a_{k-i+1})$, we have the same conclusion that for all $j=1,\dots,k-i+1$,  $a_j[x_{i-1}]=1$. 

If $i=2$ then the edge between $\alpha=\gamma_1$ and $\gamma$ implies that for all $j=1,\dots k-1$ we have $a_j[x_1]=1$. So, regardless of $i$, we have that for all $j=1,\dots,k-i+1$, $a_j[x_{i-1}]=1$.

Symmetrically, suppose $i\not=k$ and let $\gamma'_2$ be the last vertex in $L_{i+1}$ that $P$ visits. Since $L_{i+1}$ does not cover $L_i$, $\gamma'_2$ either contains $a'_{k-i+2}$ or $x$ in its representation. However, since $P_2$ does not visit $A$, Lemma~\ref{lem:first} implies that $\gamma'_2$ has the suffix $(a_{k-i+2},\dots,a_k,y)$ for some $y$. Thus, either $a_{k-i+2}=a'_{k-i+2}$ or $\gamma'_2$ contains $x$ in its representation. If $a_{k-i+2}=a'_{k-i+2}$, then $\gamma_2$ has the suffix $(a_{k-i+2},\dots,a_k,x)$, so for all $j=k-i+2,\dots,k$ we have $a_j[x_{i-1}]=1$. If $\gamma'_2$ contains $x$ in its representation, then since $\gamma'_2$ has the suffix $(a_{k-i+2},\dots,a_k,x)$, we have the same conclusion that for all $j=k-i+2,\dots,k$ we have $a_j[x_{i-1}]=1$. 

If $i=k$ then the edge between $\gamma$ and $\gamma_2=\beta$ implies that for all $j=2,\dots k$ we have $a_j[x_{k-1}]=1$. So, regardless of $i$, we have that for all $j=k-i+2,\dots k$, $a_j[x_{i-1}]=1$.

Thus, we have shown that for each $j=1,\dots k$, we have $a_j[x_{i-1}]=1$. This contradicts the fact that $a_1,\dots,a_k$ are orthogonal, completing the proof.
\end{proof}

\begin{lemma}\label{lem:excess}
Let $i$ and $j$ be such that $2\leq i\leq j \leq k$, neither $L_{i-1}$ nor $L_{j+1}$ are in a loop, and for all $\ell=i,\dots,j$, $L_\ell$ is in a loop. Let $c$ be the total number of layers that are covered by $L_i,\dots,L_j$. The subpath $P'$ of $P$ from the first time $P$ visits $L_i$ to the last time $P$ visits $L_j$ is of length at least $2(j-i)+c+1$.
\end{lemma}
\begin{proof}
We will show the contrapositive: for any $c'$, if $P'$ is of length at most $2(j-i)+c'$, then $L_i,\dots,L_j$ cover a total of at most $c'-1$ layers.

First, we note that since neither $L_{i-1}$ nor $L_{j+1}$ are in a loop, Observation~\ref{obs:layer} implies that $P'$ is entirely contained in levels $i,\dots,j$, and that no vertex on $P\setminus P'$ is in a level $i,\dots,j$. Then, since $L_\ell$ is in a loop for all $\ell=i,\dots,j$, $P'$ must contain at least two vertices from each such $L_\ell$. Let $Z$ be a subset of the vertices of $P'$ formed by taking exactly two arbitrary vertices on $P'$ from each such layer $L_\ell$. The size of $Z$ is exactly $2(j-i)+2$.

By the definition of cover, any layer can only cover the layers at most two above and below itself, so the layers $L_i,\dots,L_j$ can only cover a total of at most 4 layers. Thus, the lemma is trivially true for $c'\geq 5$. The lemma is also trivially true for $c'\leq 1$. We will condition on the length of $P'$; that is, whether $c'=2$, $3$, or $4$.

\paragraph{Case 1: $c'=2$.} $P'$ contains $2(j-i)+2$ edges and $2(j-i)+3$ vertices, so $P'$ contains exactly one vertex $\gamma$ in addition to $Z$. Since $Z\subseteq L$ and all paths from $L$ to $B$ as well as all paths from $A$ to $L$ go through $L'$, $\gamma\not\in A\cup B$. Thus, the only way that a layer in $L_i,\dots,L_j$ can cover another layer is if there exist two vertices on $P$ with different coordinate arrays (by condition 3 in the definition of cover). Every edge with both endpoints in $L$ is a swap edge, which by definition connects two vertices with the same coordinate array. Thus, $\gamma\not\in L$. We have shown that $\gamma\not\in A\cup B\cup L$, so $\gamma\in L'$. 

Let $\ell$ be such that $\gamma\in L'_\ell$. Let $\gamma_1$ and $\gamma_2$ be the vertices right before and right after $\gamma$ on $P$, respectively. Since $L_\ell$ contains the only vertices in $L$ that are adjacent to some vertex in $L'_\ell$, $\gamma_1$ and $\gamma_2$ are both in $L_\ell$. Let $\gamma_1=(a'_1,\dots, a'_{k-i},a'_{k-i+3},\dots ,a'_{k},x_{\gamma_1})$ and let $\gamma_2=(a''_1,\dots, a''_{k-i},a''_{k-i+3},\dots ,a''_{k},x_{\gamma_2})$. Since $P'$ contains no other vertices in $L_\ell$ besides $\gamma_1$ and $\gamma_2$, and $P\setminus P'$ contains no vertices in $L_\ell$, $\gamma_1$ and $\gamma_2$ are the first and last vertices from $L_\ell$ that $P$ visits. Thus, by definition, if $L_{\ell}$ covers $L_{\ell-1}$ then $a'_{k-i+3}\not=a''_{k-i+3}$ and $x_{\gamma_1}\not=x_{\gamma_2}$. Similarly, if $L_{\ell}$ covers $L_{\ell+1}$ then $a'_{k-i}\not=a''_{k-i}$ and $x_{\gamma_1}\not=x_{\gamma_2}$. Thus, if $L_\ell$ covers both $L_{\ell-1}$ and $L_{\ell+1}$, then $a'_{k-i+3}\not=a''_{k-i+3}$, $a'_{k-i}\not=a''_{k-i}$, and $x_{\gamma_1}\not=x_{\gamma_2}$. However, there are only two edges on $P'$ between $\gamma_1$ and $\gamma_2$ and each of them is either a vector-change edge or a coordinate-change edge. This means that only two out of the three above non-equalities can be true. Thus, $L_\ell$ can only cover at most one level. Since the edges on $P'$ that are not incident to $\gamma$ are swap edges and do not change the coordinate array, no other layer in $L_i,\dots,L_j$ except for $L_\ell$ covers any layer. Thus, we have shown that $L_i,\dots,L_j$ cover a total of at most one layer, as desired. 

\paragraph{Case 2: $c'=3$.} $P'$ contains $2(j-i)+3$ edges and $2(j-i)+4$ vertices, so $P'$ contains exactly two vertices in addition to $Z$. We would like to show that $L_i,\dots,L_j$ cover a total of at most two layers. The only way for $L_i,\dots,L_j$ to cover more than two layers is to use the last part of the definition of cover, which allows $L_4$ or $L_{k-2}$ to cover two layers. Thus, we will show that if $i=4$ and $L_4$ covers $L_3$ and $L_2$, or if $j=k-2$ and $L_{k-2}$ covers $L_{k-1}$ and $L_k$, then no other layers can be covered by $L_i,\dots,L_j$.

Suppose $i=4$ and $L_4$ covers $L_3$ and $L_2$. By the definition of cover, $P$ visits $A$. Then, since all vertices in $P\cap L_4$ are on $P'$, Observation~\ref{obs:ab} implies that $P'$ visits $A$. Therefore, the two vertices on $P'$ in addition to $Z$ are one vertex in $A$ and one vertex in $L'_4$. Thus, the only edges that $P'$ uses are swap edges with endpoints in $L_4\cup\dots\cup L_j$, an $a$-type back edge from $L_4\cup\dots\cup L_j$ to $A_4\cup\dots\cup A_j$, an $a$-type back edge from $A_4\cup\dots\cup A_j$ to $L'_4$, and either a vector-change or coordinate-change edge from $L'_4$ to $L_4$. It will be important to note that by construction, all of these edges go from a vertex that contains a vector $a'_{k-j}\in S_{k-j}$ to another vertex containing the same vector $a'_{k-j}\in S_{k-j}$. This is because the only edges that change a vector $S_{k-j}$ to a different vector in $S_{k-j}$ are vector-change edges between $L_j$ and $L'_j$ or between $L_{j+3}$ and $L'_{j+3}$, which $P'$ does not use.

 Our goal is to show that $L_j$ does not cover $L_{j+1}$.  Let $\gamma_1= (a''_1,\dots, a''_{k-j},a''_{k-j+3},\dots ,a''_{k},x_{\gamma_1})$ and $\gamma_2=(a'''_1,\dots, a'''_{k-j},a'''_{k-j+3},\dots ,a'''_{k},x_{\gamma_2})$ be the two vertices in $L_j$ that $P'$ visits. Since $P\setminus P'$ contains no vertices in $L_j$, $\gamma_1$ and $\gamma_2$ are the first and last vertices from $L_j$ that $P$ visits. Thus, if $L_j$ covers $L_{j+1}$, then $a''_{k-j}\not=a'''_{k-j}$. But we have already shown that this cannot be the case since every edge on $P'$ does not change the vector $a'_{k-j}\in S_{k-j}$.
 
 The $j=k-2$ case is completely symmetric to the $i=4$ case, but we include it for completeness. Suppose $j=k-2$ and $L_{k-2}$ covers $L_{k-1}$ and $L_k$. This means that $P$ visits $B$. Then, since all vertices in $P\cap L_{k-2}$ are on $P'$, Observation~\ref{obs:ab} implies that $P'$ visits $B$. Therefore, the two vertices on $P'$ in addition to $Z$ are one vertex in $L'_{k-2}$ and one vertex in $B$. Thus, the only edges that $P'$ uses are swap edges with endpoints in $L_i\cup\dots\cup L_{k-2}$, either a vector-change or coordinate-change edge from $L_{k-2}$ to $L'_{k-2}$, a $b$-type back edge from $L'_{k-2}$ to $B_i\cup\dots\cup B_{k-2}$, and a $b$-type back edge from $B_i\cup\dots\cup B_{k-2}$ to $L_i\cup\dots\cup L_{k-2}$. It will be important to note that by construction, all of these edges go from a vertex that contains a vector $a'_{k-i+3}\in S_{k-i+3}$ to another vertex containing the same vector $a'_{k-i+3}\in S_{k-i+3}$.

 Our goal is to show that $L_i$ does not cover $L_{i-1}$.  Let $\gamma_1= (a''_1,\dots, a''_{k-i},a''_{k-i+3},\dots ,a''_{k},x_{\gamma_1})$ and $\gamma_2=(a'''_1,\dots, a'''_{k-i},a'''_{k-i+3},\dots ,a'''_{k},x_{\gamma_2})$ be the two vertices in $L_i$ that $P'$ visits.  Since $P\setminus P'$ contains no vertices in $L_i$, $\gamma_1$ and $\gamma_2$ are the first and last vertices from $L_i$ that $P$ visits. Thus, if $L_i$ covers $L_{i-1}$, then $a''_{k-i+3}\not=a'''_{k-i+3}$. But we have already shown that this cannot be the case since every edge on $P'$ does not change the vector $a'_{k-i+3}\in S_{k-i+3}$.
 
 \paragraph{Case 3: $c'=4$.} $P'$ contains $2(j-i)+4$ edges and $2(j-i)+5$ vertices, so $P'$ contains exactly three vertices in addition to $Z$. We would like to show that $L_i,\dots,L_j$ cover a total of at most three layers. Suppose for contradiction that $L_i,\dots,L_j$ cover four layers. The only way for this to happen is if $L_4=L_i$ covers $L_3$ and $L_2$, and $L_{k-2}=L_j$ covers $L_{k-1}$ and $L_k$. This means that $P$ visits both $A$ and $B$. Then, since all vertices in $P\cap L_{k-2}$ and all vertices in $P\cap L_4$ are on $P'$, Observation~\ref{obs:ab} implies that $P'$ visits both $A$ and $B$. Since the only edges to $B$ from another set are from $L'_{k-2}$ and the only edges from $A$ to another set are to $L'_4$, $P'$ must contain a vertex in each of $A$, $B$, $L'_{k-2}$, and $L'_4$. However, $P'$ only contains three vertices in addition to $Z$, a contradiction.
 \end{proof}
 
 We will now take the sum over all of the subpaths defined by Lemma~\ref{lem:excess}. Form a partition $\mathcal{P}$ of $2,\dots, k$ where each piece of $\mathcal{P}$ contains a maximal interval $i,\dots, j$ such that $L_i,\dots, L_j$ are all in loops, as well as the layers that $L_i,\dots, L_j$ cover. (The maximality condition means that either $i=2$ or $L_{i-1}$ is not in a loop, and either $j=k$ or $L_{j+1}$ is not in a loop.) To make this a true partition of $2,\dots,k$, for any $\ell$ that has been placed in two pieces of the partition due to being covered by multiple layers, we remove $\ell$ from an arbitrary one of these two pieces. Now, by Lemma~\ref{lem:cover}, every value $2,\dots, k$ is in exactly one piece of $\mathcal{P}$. 
 
Consider a piece $i,\dots,j$ of $\mathcal{P}$, and let $i'\geq i$ and $j'\leq j$ be such that $L_{i'},\dots,L_{j'}$ are each in a loop and the remaining layers in $L_i,\dots,L_j$ are covered. Let $c$ be the number of covered layers in $L_i,\dots,L_j$. That is, $j'-i'+c=j-i$.  By Observation~\ref{obs:layer} and the maximality condition of $\mathcal{P}$, any path with both endpoints in $L'_i,\dots,L'_j$ is entirely contained within levels $i',\dots, j'$. 

Let $P_{i,j}$ and $P_{i',j'}$ be the subpaths of $P$ in the graph induced by levels $i,\dots, j$ and levels $i',\dots,j'$, respectively. By Lemma~\ref{lem:excess}, $P_{i',j'}$ is of length at least $2(j'-i')+c+1$. 
Adding an edge to each of the $c$ covered levels, $P_{i,j}$ is of length at least $2(j'-i'+c)+1=2(j-i)+1$. 

Let $p$ be the number of pieces of $\mathcal{P}$. We first calculate the sum over the lengths of all $P_{i,j}$. Since $\mathcal{P}$ partitions $2,\dots,k$, we have $\sum_{(i,\dots,j)\in\mathcal{P}} 2(j-i)+1=2(k-1-p)+p=2k-p-2$. In addition to the edges in each $P_{i,j}$, $P$ also contains at least $p-1$ edges such that each endpoint is in a different piece of $\mathcal{P}$, as well as an edge from $L_1$ to $L_2$ and an edge from $L_k$ to $L_{k+1}$. Thus, the length of $P$ is at least $2k-p-2+(p-1)+2=2k-1$.

\section{The $k=4$ case}
\label{sec:k=4}
We mainly use the construction of \cite{4ov} with a few alterations, to get a reduction from $4$-OV to directed \emph{unweighted} diameter. Note that the lower bound of \cite{4ov} is for directed weighted diameter. 

Given a $4$-OV instance $S$ where each vector in $S$ is of length $d$, that is, there are $d$ coordinates, we create a graph $G$ such that if $S$ is a NO case then the diameter of $G$ is at most $4$ and if $S$ is a YES case the diameter of $G$ is at least $7$. See Figure \ref{fig:k4}.

We make $4$ copies of $S$ and call them $S_1,\ldots,S_4$. The vertex set of our graph is essentially the same as \cite{4ov}, and we redefine it for completeness. The graph $G$ consists of layers $L_i$ for $i=1,\ldots,5$. Vertices of $L_1$ are $3$ tuples of the form $(a_1,a_2,a_3)$, where $a_i\in S_i$. Vertices of $L_5$ are $3$ tuples of the form $(b_2,b_3,b_4)$, where $b_i\in S_i$. Vertices of $L_2$ are of the form $(a_1,a_2,x)$, vertices of $L_3$ are of the form $(a_1,b_4,x)$ and vertices of $L_4$ are of the form $(b_3,b_4,x)$, where $a_i,b_i\in S_i$ and $x$ is a coordinate array of length $3$ satisfying the conditions of Table \ref{tab:Li}. We have an additional layer $L'_3$ with vertices of the form $(a_1,b_4,x)$ for every coordinate array of length 3, 
where at least $5$ out of the $6$ following conditions hold: $a_1[x_{\ell}]=1$ for each $\ell$ and $b_4[x_{\ell}]=1$ for each $\ell$.\footnote{In fact, this last constraint is not necessary and we can instead define $L'_3$ to be all vertices of the form $(a_1,b_4,x)$. We include this last constraint to be consistent with~\cite{4ov}, so that we can use the correctness of their construction as a black box to argue the correctness of our construction.} 
Finally, we have two vertices $v$ and $u$.

\begin{figure}[h]
    \centering
    \includegraphics[width=0.6\linewidth]{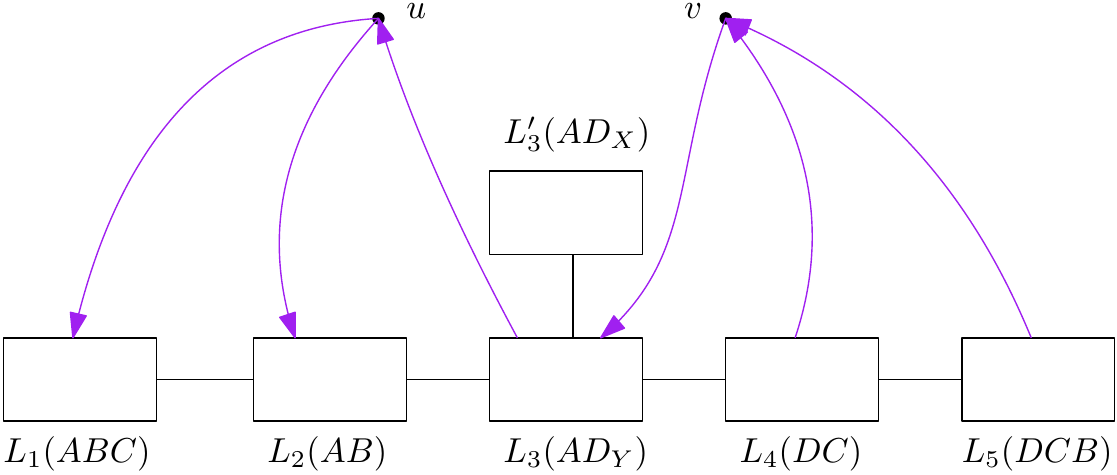}
    \caption{$k=4$ construction. The names in parentheses are from the construction of \cite{4ov} and are put here for ease of comparison.  The purple edges are the fixed edges.}
    \label{fig:k4}
\end{figure}

We have swap edges between $L_i$ and $L_{i+1}$ for $i=1,\ldots,4$. We have vector-change edges between $L_3$ and $L'_3$, between $(a_1,b_4,x)\in L_3$ and $(a'_1,b'_4,x)\in L'_3$ where either $a_1=a'_1$ or $b_4=b'_4$. We have coordinate-change edges between $L_3$ and
$L'_3$, and within $L'_3, L_2$ and $L_4$.
Finally we have fixed edges connected to $u$ and $v$ as follows. Every node in $L_4\cup L_5$ has a directed edge to $u$, and every node in $L_3$ has a directed edge from $u$. Similarly, every node in $L_1\cup L_2$ has an edge from $v$, and every node in $L_3$ has an edge to $v$.

Note that the only difference between our construction and the construction of \cite{4ov} is the coordinate-change edges inside $L'_3$, and the fixed edges.
\subsection{NO instance of $4$-OV implies diameter at most 4}

\subsubsection{Using fixed edges}
\paragraph{Case 1: $v$ to $\beta\in V\setminus\{v\}$.} If $\beta\in L_1\cup L_2$, there is a direct edge from $v$ to $\beta$. Now note that for any $i$, any node in $L_i$ has a swap edge to a node in $L_{i+1}$ and $L_{i-1}$ changing one of its vectors to all 1s vector. So there is an undirected path of length at most $3$ between  $\beta\in L_3\cup L_4\cup L_5$ and some node $\beta'$ in $L_2$.
There is a direct edge from $v$ to $\beta'$, so $d(v,\beta)\le 4$. If $\beta\in L'_3$, there is path of length $2$ from a node $\beta'\in L_2$, using a swap edge and a coordinate-change edge (without changing the coordinate), and hence $d(v,\beta)\le 3$. Finally if $\beta=u$, then from $v$ we take a fixed edge to some node in $L_2$, then we take two swap edges to some node in $L_4$, and then we take a fixed edge to $u$. 
\paragraph{Case 2: $\alpha \in V\setminus\{v\}$ to $v$.} We need to show that there is a path of length at most $3$ from $\alpha$ to some node $\alpha'\in L_3$. Then since $d(\alpha',v)=1$, we have a path of length $4$ from $\alpha$ to $v$. Again note that for any $i$, any node in $L_i$ has a swap edge to a node in $L_{i+1}$ and $L_{i-1}$ changing one of its vectors to all 1s vector. So for $\alpha\in L_1,L_2,L_4,L_5$, there is a path of length $2$ from $\alpha$ to some node in $L_3$. If $\alpha=(a_1,b_4,x)\in L'_3$, then we can take a coordinate change to $(a_1,b_4,x')\in L_3$ where $x'[i]=C(a_1,b_4)$ for $i=1,2,3$. Finally if $\alpha=u$, there is a direct edge from $u$ to all nodes in $L_3$. 
\paragraph{Case 3: $u$ to $\beta\in V\setminus\{u\}$.} Symmetric to case 2.
\paragraph{Case 4: $\alpha \in V\setminus\{u\}$ to $u$.} Symmetric to case 1.
\paragraph{Case 5: $\alpha\in V$ to $\beta\in L_1\cup L_2$.} From $\alpha$ take at most two edges to some node in $L_3$. The we take a fixed edge to $v$, and finally we take another fixed edge to $\beta$. 
\paragraph{Case 6: $\alpha\in L_4\cup L_5$ to $\beta \in V$.} Symmetric to case 5.

\subsubsection{Using variable edges}

In a NO instance, for every set $F$ of at most $4$ vectors, there exists a coordinate that is 1 for every vector in $F$. Given a set $F$ of at most $4$ vectors, recall that $C(F)$ denotes a coordinate that is 1 for every vector in $F$. 

The only remaining cases that are not covered by fixed edges are the following. 
\paragraph{Case 1: $\alpha\in L_1\cup L_2$ to $\beta\in L_4\cup L_5$.} Let $\alpha=(a_1,a_2,a_3)$ if $\alpha\in L_1$ and let $\alpha=(a_1,a_2,x_{\alpha})$ if $\alpha\in L_2$, in which case $a_3$ is the all 1s vector. Let $\beta=(b_2,b_3,b_4)$ if $\beta\in L_5$ and let $\beta=(b_3,b_4,x_{\beta})$ if $\beta\in L_4$, in which case $b_2$ is defined to be the all 1s vector. We will use the coordinate array $x=(x_1,x_2,x_3)$ where for each $i=1,2,3$, $x_i=C(a_1,\ldots,a_{4-i},b_{5-i},\ldots,b_4)$. From $\alpha$, take a swap edge or a coordinate-change edge to $(a_1,a_2,x)\in L_2$, a swap edge to $(a_1,b_4,x)\in L_3$, a swap edge to $(b_3,b_4,x)\in L_4$, and then a swap edge or a coordinate-change edge to $\beta$.
\paragraph{Case 2: $\alpha\in L_1\cup L_2$ to $\beta\in L_3\cup L'_3$.} Let $\alpha=(a_1,a_2,a_3)$ if $\alpha\in L_1$ and let $\alpha=(a_1,a_2,x_{\alpha})$ if $\alpha\in L_2$, in which case $a_3$ is the all 1s vector. Let $\beta=(a'_1,b_4,x_{\beta})$. Consider the following coordinate array $x$: $x_1=C(a_1,a_2,a_3,b_4)$, and $x_i=C(a_1,a_2,a'_1,b_4)$ for $i=2,3$. We take the following path: From $\alpha$, take a coordinate-change edge or a swap edge to $(a_1,a_2,x)\in L_2$. Then take a swap edge to $(a_1,b_4,x)\in L_3$, a vector-change edge to $(a'_1,b_4,x)\in L'_3$, and finally a coordinate-change edge to $\beta$. 
\paragraph{Case 3: $\alpha\in L_3\cup L'_3$ to $\beta\in L_4\cup L_5$.} Symmetric to Case 3.
\paragraph{Case 4: $\alpha\in L_3\cup L'_3$ to $\beta\in L_3\cup L'_3$.} Let $\alpha = (a_1,b_4,x_{\alpha})$ and $\beta=(a'_1,b'_4,x_{\beta})$. Consider the coordinate array $x$ where $x_i=C(a_1,a'_1,b_4,b'_4)$ for $i=1,2,3$. We use the following path: From $\alpha$ take a coordinate-change edge to $(a_1,b_4,x)\in L'_3$. Then go to $(a'_1,b_4,x)\in L_4$ and then to $(a'_1,b'_4,x)\in L'_4$ using two vector-change edges, and finally take a coordinate-change edge to $\beta$.

\subsection{YES instance of $4$-OV implies diameter at least 7}
Suppose that $a_1,\ldots,a_4$ are orthogonal. We will show that the distance from $\alpha=(a_1,a_2,a_3)$ to $\beta=(a_2,a_3,a_4)$ is at least $7$ if it uses one of the following edges: a coordinate-change edge in $L'_3$ or a fixed edge. This is because the rest of the construction is included in the construction of \cite{4ov}, and if the path does not use any of these edges, it is included in the construction of \cite{4ov} and thus it is of length at least $7$.

First suppose that the path uses a coordinate-change edge inside $L'_3$. Then the path has at least two nodes in $L'_3$ and $L_3$, and at least one node in each $L_i$ for $i=1,2,4,5$. So the path is of length at least $7$. 

Next suppose that the path uses a fixed edge. So the path passes through $u$ or $v$. This means that the path has a subpath passing through $L_j$, then $v$ or $u$, then $L_i$, where $j>i$. So $L_i$ and $L_j$ have at least two nodes in the path, the path has at least one node in each $L_r$ for $r=1,\ldots,5$, $r\neq i,j$ and it has $u$ or $v$. So it has at least $8$ nodes and hence it is of length at least $7$.

\hspace{5mm}

\paragraph{Acknowledgments} We would like to thank Virginia Vassilevska Williams for discussions and encouragement.

\bibliographystyle{alpha}
\bibliography{bib} 
\end{document}